\documentclass[english,12pt]{article}
\usepackage{lmodern}

\usepackage[T1]{fontenc}
\usepackage[dvipsnames,svgnames,x11names,hyperref]{xcolor}
\usepackage[margin=1truein]{geometry}
\usepackage{tabularray}
\usepackage{babel}
\usepackage{booktabs}
\usepackage{amsmath}
\usepackage{amsthm}
\usepackage{amssymb}
\usepackage{graphicx}
\usepackage{setspace}
\usepackage{caption}
\usepackage[lined,ruled,boxed,commentsnumbered]{algorithm2e}

\usepackage{subcaption}
\usepackage{comment}

\usepackage{natbib}
\setcitestyle{authoryear,open={(},close={)}}

\usepackage{tikz}
\usepackage{soul}
\usepackage{lineno}

\usepackage{dsfont}
\usepackage{lipsum}

\newcommand{\secname}{Section}
\newcommand{\algname}{Algorithm}

\usepackage[unicode=true,bookmarks=true,bookmarksnumbered=false,bookmarksopen=false,breaklinks=false,pdfborder={0 0 0},pdfborderstyle={},backref=page,colorlinks=true]{hyperref}
\hypersetup{linkcolor=RoyalBlue,citecolor=RoyalBlue, urlcolor = RoyalBlue}

\makeatletter

\newcommand{\pr}{\mbox{pr}}
\newcommand{\logL}{L}
\newcommand{\HR}{\varphi}
\newcommand{\indicator}{\mathds{1}}

 \theoremstyle{remark}

 \theoremstyle{definition}
 
 \newtheorem{prop}{\protect Proposition}
  \newtheorem{assumption}{\protect Assumption}

\newcommand*\samethanks[1][\value{footnote}]{\footnotemark[#1]}

\title{A tree-based scan statistic for database studies with time-to-event outcomes}
\date{}

\author{Massimiliano Russo\thanks{
Department of Statistics, The Ohio State University, Columbus, OH, U.S.A., russo.325@osu.edu},
Georg Hahn\thanks{Division of Pharmacoepidemiology and Pharmacoeconomics, Brigham and Women's Hospital and Harvard Medical School, Boston, MA, U.S.A.},
Krista F. Huybrechts\samethanks,
and Shirley V. Wang\samethanks}

\begin{document}

\maketitle
\begin{abstract}\setstretch{1.0}
Tree-based scan statistics (TBSSs) are machine learning methods for disproportionality analyses. They simultaneously scan for thousands of hierarchically related health outcomes to detect potential signals of harm from drugs and vaccines, controlling for multiplicity. TBSSs have been extensively used to mine insurance claims databases to detect potential drug adverse events. Current TBSS implementations do not allow comparative safety evaluations with time-to-event outcomes. Explicitly accounting for event timing in analyses can improve power to detect signals compared to methods that use only event counts. We propose three TBSS methods that explicitly leverage event timing to detect potentially harmful effects of drugs. The first assumes proportional hazard rates for each node of the outcome hierarchy and uses a permutation scheme for inference. The second builds on exponential survival models for the terminal nodes of the hierarchy, assuming constant hazard rates at each node, and uses a parametric bootstrap for inference. The third uses robust asymptotic approximations of the hazard rates in connection with an approximate parametric bootstrap.  We compare the proposed methods with standard event count based TBSSs in simulation scenarios. Finally, we present results from a database study comparing two glucose-lowering medications among adults with type 2 diabetes. 
\end{abstract}

\section{Introduction}
Tree-based scan statistics (TBSSs) are methods to conduct disproportionality analyses with many hierarchically related outcomes while maintaining tight control of type-I error~\citep{kulldorff:2003}. They are based on scan statistical theory~\citep{glaz:2001}, and search for potential increased risks of drugs and vaccines on thousands of hierarchically related outcomes, sequentially aggregating these outcomes and computing disproportionality measures for each aggregation. They have been routinely used to implement safety surveillance by organizations such as the Food and Drug Administration (FDA) and the Centers for Disease Control and Prevention (CDC); see, for example,  \citet{brown:2013},  \citet{yih:2023a}, and \citet{yih:2023b}.  TBSS are typically used to generate hypotheses on the safety of medical products and facilitate resource allocation by prioritizing areas for further investigation. 

For studies conducted using routinely collected healthcare data, such as large insurance claims databases or electronic health records, the hierarchical structure used by TBSSs to relate outcomes typically reflects a hierarchical classification of clinical diagnosis codes. Common structures used in pharmacoepidemiology include the 9th and 10th iterations of the International Disease Classification (ICD) taxonomy~\citep[e.g.,][]{huybrechts:2021,fralick:2021}. An example of a hierarchical structure based on the 10th revision of the ICD is depicted in \figurename~\ref{fig:ex_tree}. Each node in \figurename~\ref{fig:ex_tree} provides a definition of a cardiovascular event (outcome) that can be organized across multiple levels of a tree, from broad outcome definitions as disease categories (top) to increasingly specific outcome-nodes (bottom). Other hierarchical structures, such as the Medical Dictionary for Regulatory Activities (MedDRA), have also been considered~\citep{russo:2025}.
\begin{figure}[ht!]
    \centering
    \resizebox{\textwidth}{!}{
\begin{tikzpicture}[sibling distance=35mm, every text node part/.style={align=center}, level distance=20mm,
   every node/.style={fill=RoyalBlue!30,rectangle, inner sep=3pt}]
\node {\textbf{I00--I99} \\ Diseases of the circulatory system }
    child {node {\textbf{I20-I25} \\
Ischemic heart diseases} child {node[fill = white] {\textbf{...}}}}
    child {node[fill=white] {\textbf{...}}}
    child {node {\textbf{I30-I5A} \\
Other forms of heart disease}
      child {node {\textbf{I50} \\ Heart failure}
child {node {\textbf{I50.2} \\ Systolic (congestive) \\  heart failure} 
 child {node {\textbf{I50.21} \\ Acute systolic \\ (congestive) \\ heart failure}} 
 child {node[fill=white] {\textbf{...}}} 
 child {node {\textbf{I50.23} \\ Acute on chronic systolic \\ (congestive) \\ heart failure}} 
 }
child {node[fill=white] {\textbf{...}}}
child {node {\textbf{I50.3}  \\ Diastolic (congestive) \\ heart failure}
child {node[fill = white] {\textbf{...}}}}
}
      child {node[fill=white] {\textbf{...}}}
      child {node {\textbf{I5A} \\
Non-ischemic myocardial injury\\ (non-traumatic)} child {node[fill = white] {\textbf{...}}}}
    };
\end{tikzpicture}}
 \caption{Example of a tree structure based on the 10th revision of the International Statistical Classification of Diseases and Related Health Problems (ICD-10). This figure depicts a section of a five-level tree that can be used to describe diseases of the circulatory system; several nodes have been removed and indicated with the `\textbf{...}' symbol. Note that codes referring to finer definitions are characterized by increasing decimal points.} 
    \label{fig:ex_tree}
\end{figure}
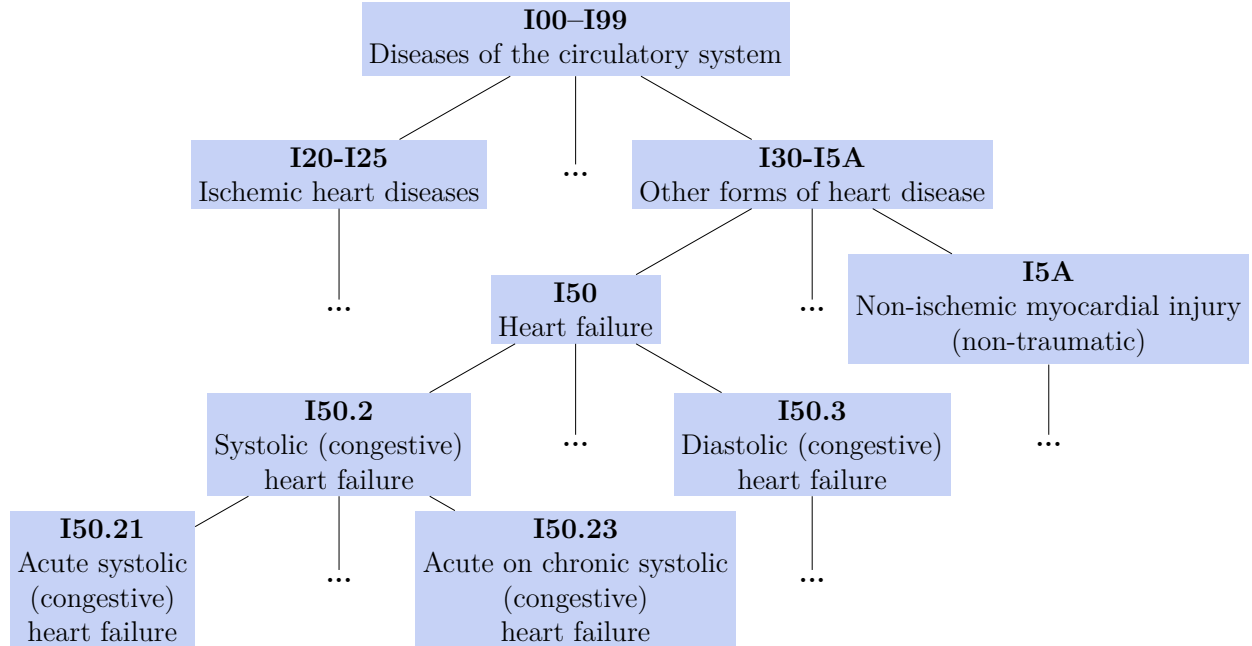
Typically, in comparative safety studies, the focus is on testing whether the distribution of a given outcome node differs between two exposure groups. TBSSs jointly consider all the outcome nodes defined by a hierarchy and test the global null hypothesis that there are no outcome nodes across the hierarchy for which there is a difference in risk between the exposure groups. The alternative hypothesis is that there is at least one outcome node for which there is a difference in risk. As a test statistic for the global null, TBSSs use the most extreme observed disproportion across all nodes in the hierarchy. The distribution of this statistic is typically not available in closed form, and it is approximated using Monte Carlo (MC) methods to adjust for multiple comparisons and account for the correlation among the outcomes~\citep{kulldorff:2003}. Node-specific hypotheses are also tested, and the corresponding p-values are used to identify the nodes for which the null hypothesis is likely to be violated. TBSS procedures are very general and require the specification of four main ingredients: 1) a hierarchical outcome structure (i.e., the tree); 2) a summary statistic that effectively quantifies disproportion; 3) an MC algorithm to generate data replicates under the global null hypothesis; and 4) observed outcome counts for the terminal nodes (input data).  We describe the TBSS procedure in detail in \secname~\ref{sec:TBSS}.

Despite this generality, to date, only a few combinations of test statistics and MC algorithms have been considered in the literature. The most common approaches are based on the Bernoulli TBSS~\citep[e.g.,][]{schachterle:2019}, which considers successes and failures, and the Poisson TBSS~\citep[e.g.,][]{kulldorff:2013}, which compares the observed number of events at each node to expected counts. Some extensions to these approaches have been considered. For example, \citet{fralick:2021} and \citet{russo:2024} introduce extensions of the Bernoulli TBSS for stratified analysis. \citet{park:2022} generalize the Poisson TBSS to include a parameter to account for the excess of zeros. Moreover,  \citet{heo:2023} consider a case-control design that accounts for dependencies within matched pairs. Although TBSS models have been extensively used in practice, to the best of our knowledge, statistical properties and error rates control beyond the rejection of the global null have not been discussed in the literature. However, it is common practice to use node-specific hypotheses to identify potential signals and inform further analysis. Our first contribution is to provide a characterization of the {\em Family-Wise Error Rate (FWER)} for the node-specific hypotheses in \secname~\ref{sec:TBSS}. 

Additionally, most databases include information on when variables (exposure, covariates, and outcomes) have been measured or recorded in the database. This information is crucial in the design phase of a study, for example, to select confounders, specify study periods,  decide on follow-up and washout windows and, more broadly, to make comparisons between drugs fair and reduce bias~\citep[e.g.,][]{gerhard:2008,suissa:2008,pottegard:2022}. At the analysis stage, when the interest is in one binary outcome and time-to-event data are available, a Cox model is typically preferred to a logistic regression as it is generally more efficient~\citep {annesi:1989}. Currently, when the interest is in studying multiple outcomes, TBSS methods use time-to-event data at the design stage to ensure that outcome events are captured within follow-up windows that are comparable between the index exposure and comparator. However, information on the actual time-to-event is discarded at the analysis stage.

To fill this gap, we propose three novel TBSSs that leverage time-to-event data to increase signal-detection power. The first method (\textit{CoxTBSS}, \secname~\ref{sec:cox}) estimates a Cox proportional hazards model~\citep{Cox:1972} for each node of the hierarchical structure, and uses a permutation scheme to approximate the distribution of the test statistics under the global null. \textit{CoxTBSS} is flexible and applicable in many settings but requires patient-level data for inference, which can be restrictive in distributed contexts and computationally prohibitive for large studies.
The second method (\textit{ExponentialTBSS}, \secname~\ref{sec:exp}) estimates constant hazard rates for the leaf-level nodes and uses a parametric bootstrap, which requires only four data summaries for each leaf node to conduct inference. The resulting approach is similar to the traditional Poisson TBSS, but it does not assume that expected counts are observed without error. Finally, the third method (\textit{RobustTBSS}, \secname~\ref{sec:robust}) robustifies the second method by using approximate inference that is robust to model misspecification, particularly the assumption of constant hazard rates in \textit{exponentialTBSS}.

The paper is structured as follows. \secname~\ref{sec:TBSS} describes the structure of TBSSs and provides conditions for controlling the Family-Wise Error Rate (FWER). \secname~\ref{sec:method} introduces our three TBSS procedures to analyze time-to-event outcomes. \secname~\ref{sec:simulations} compares our proposed procedures to the current TBSS implementation in realistic simulation settings. \secname~\ref{sec:application} considers a database study that compares exposure to two classes of glucose-lowering medications in adults with type 2 diabetes. \secname~\ref{sec:discussion} concludes with a brief discussion.


\section{Tree-based scan statistics}\label{sec:TBSS}
TBSSs simultaneously analyze hierarchically related outcomes. The input data for TBSSs are usually data summaries (e.g., the sufficient statistics of a parametric model) defined at the lowest level of the hierarchy of interest. We refer to this hierarchy as a `tree', and the nodes in the lower level as terminal nodes or leaves of the tree. In general, this tree can be any directed acyclic graph that specifies how the data summaries can be aggregated meaningfully. Formally, the tree can be defined by the tuple $\mathcal T = (\mathcal G, \mathcal E)$, where $\mathcal{G}$ is a set of nodes connected by directed edges $\mathcal E$. The edges point from finer to coarser scales, i.e., from the leaves toward the root of the tree. The set of leaves, or terminal nodes, $\mathcal L \subseteq \mathcal G,$ is the set of all the nodes that do not have edges pointing at them. For each node $g \in \mathcal G$, we also define the set $\mathcal L_g \subseteq \mathcal L$  of all leaves connected, directly or through other nodes, to the node $g.$

For a generic TBSS method, data summaries $D^{(\ell)}$ are collected at the leaf level of the tree for each $\ell \in \mathcal L$. Their dimension and types depend on the specific TBSS analysis. Importantly, data for different leaves are considered realizations of independent random variables, and the data summaries $D^{(\ell)}$ are typically sufficient statistics of the parametric model for the leaf-level data. For example, if we assume that the leaves are distributed as independent Poisson random variables~\citep[as in][]{huybrechts:2021}, we typically have observed events ($O_\ell$), e.g., number of events for an index drug, and expected events ($E_\ell$), e.g.,  number of events for a comparator drug, as summaries for each leaf, meaning $D^{(\ell)} = (O_\ell, E_\ell)$. The data stored in non-leaf nodes are defined by `aggregating' the data of the leaves that are connected to that node. In the Poisson case, we have $D^{(g)} =  (O_g, E_g)$, where $O_g  = \sum_{\ell \in \mathcal{L}_g}O_\ell$ and $ E_g = \sum_{\ell \in \mathcal{L}_g} E_\ell.$  Note that in the classical Poisson TBSS, the number of expected counts for each node of the tree ($E_g$ for $g \in \mathcal{G}$) is considered fixed.  In practice, we often define the counts using only incident events; see \secname~\ref{sec:incident} in the supplementary materials for discussion.

TBSSs test the global null hypothesis that there are no outcome nodes in the hierarchical tree for which the number of events for the index drug differs from that of the comparator drug. For example,  in our Poisson example, we can test if there are nodes where the average number of observed counts differs from the expected counts. In other words, we test the null hypothesis $H_0 : \{\mathbb E[O_g] = E_g, \mbox{ for all } g \in \mathcal{G}\}$ against the alternative $H_1 : \{\mathbb E[O_{g'}] \neq E_{g'}, \mbox{ for some } g' \in \mathcal{G}\}.$  We also use the notation $H_{0g}=\{\mathbb E[O_g] = E_g\}$ and $H_{1g} = \{\mathbb E[O_{g}] \neq E_{g}\}$ to indicate node-specific hypotheses.

To test the global null $H_0,$ for each node $g \in \mathcal G$, TBSSs first compute the statistics $T_g = T_g (D^{(g)})$ that are functions of the data for the corresponding nodes. Typically, $T_g$ is based on the log-likelihood ratio statistic (LLR) of a parametric model for the data defined at leaf-level nodes, but other statistics---measuring the discrepancy between $H_{0g}$ and $H_{1g}$---can be used.   In the Poisson example, we can use the LRT
$T_g = (E_g -O_g) + O_g \times  (\log O_g - \log E_g ).$  The test statistic for $H_0$ is the maximal observed disproportion in the tree,   $T = \max_{g \in \mathcal G} T_g.$ The distribution of $T$ under $H_0$ is typically not available in closed form but can be approximated via Monte Carlo (MC) simulation. For Poisson TBSS~\citep{kulldorff:2003}, the null distribution of  $T$ is approximated via simulation under the global null. In other words, data replicates are obtained by sampling leaf-level data summaries under $H_0$ as $O_\ell \overset{H_0}{\sim} \mbox{Poisson}(E_\ell), $ for $\ell \in \mathcal L,$ and then aggregated using the hierarchy $\mathcal{T}.$

The p-value for the null hypothesis $H_0$ is defined as the tail area probability ${\tt pv} =\mbox{pr}(T > T^{\text{obs}} \mid H_0),$ where $T^{\text{obs}}$ is the statistic computed on the observed data. In addition, node-specific p-values are defined via ${\tt pv}_{g} = \mbox{pr}(T > T_g^{\text{obs}} \mid H_0),$ where $T_g^{\text{obs}}$ is the statistic $T_g$ computed on the observed data for node $g \in \mathcal{G}$. By definition,  ${\tt pv} = \min_{g \in \mathcal G} \{{\tt pv}_g\},$ i.e., the p-value for $H_0$ is the smallest of the node-specific p-values.  Moreover, the p-values $ \{{\tt pv}_{g}\}_{g \in \mathcal G}$ can be used to test if, for a particular node, the number of observed events is significantly different (or larger) than the expected number of events.  We can define as {\em statistical alerts} all nodes with a p-value lower than a pre-specified threshold $\alpha$---typically $0.05$. Alerts often include nodes that share common leaves, meaning they are on the same tree path.

\subsection{Family-wise error rate control}
For TBSSs, weak control of the FWER---i.e., the probability of having a false positive result under the global null hypothesis \citep[e.g.,][Section 1.2.1]{dickhaus:2014}---directly follows from the definition of p-values. However, to the best of our knowledge, the literature has not discussed whether and how the algorithm controls for false positives in node-specific tests.

A challenge in characterizing the FWER of TBSSs is that the signal propagates across the tree. How this propagation happens depends on which node-specific hypotheses are considered. As a general case, we consider two random variables $O_g$ and $E_g$  for each node  $g \in \mathcal{G}$, and we are interested in testing $O_g \overset{\text{d}}{=} E_g,$ where `$\overset{\text{d}}{=}$' indicates equality in distribution. To illustrate signal propagation, consider an example with three nodes $\mathcal{G} = \{1,2,3\}$, whereby the leaves $1$ and $2$ are connected to node $3$, i.e., $O_3 = O_1 + O_2,$ and $E_3= E_1 + E_2$. We differentiate four cases:
\begin{enumerate}
\item $O_1 \overset{\text{d}}{=} E_1$ and $O_2 \overset{\text{d}}{=} E_2:$ This  implies that $O_3 \overset{\text{d}}{=} E_3.$
\item $O_1 \overset{\text{d}}{=} E_1$ and $O_2 \overset{\text{d}}{\neq} E_2:$ This implies that $O_3 \overset{\text{d}}{\neq} E_3.$
\item   $O_1 \overset{\text{d}}{\neq} E_1$ and $O_2 \overset{\text{d}}{\neq} E_2,$ but  either
    $O_1 \overset{\text{d}}{\neq} E_2$ or $O_2 \overset{\text{d}}{\neq} E_1:$ This implies that $O_3 \overset{\text{d}}{\neq} E_3.$
    \item $O_1 \overset{\text{d}}{\neq} E_1$ and $O_2 \overset{\text{d}}{\neq} E_2,$ but  $O_1 \overset{\text{d}}{=} E_2$ and $O_2 \overset{\text{d}}{=} E_1:$ This implies that $O_3 \overset{\text{d}}{=} E_3.$ 
\end{enumerate}
Cases 1--3 are intuitive as they follow intersection and union operations: If node 3 is a parent node of only nodes that are under the null, then node 3 is under the null. In contrast, if at least one of the child nodes of node 3 is under the alternative, then node 3 is also under the alternative. Case 4 is more problematic. In our setting that compares the safety profiles of two drugs, case 4 implies that the index drug is as protective (harmful) for node one as the comparator drug is harmful (protective) for node two with the same strength, and vice versa.  This symmetry cancels the two signals in nodes one and two. Case 4 is possible but arguably unlikely in many disease settings.

We restrict our attention to trees $ \mathcal{T}$ where only cases 1--3 are possible. Equivalently,  we consider FWER control only on sets $\mathcal{I}_0$ following union and intersection operations.  To characterize the set $\mathcal I_0$, we first define its complement $\mathcal I_1$ that includes the nodes such that the alternative hypothesis holds. Since we exclude case 4, we assume signals propagate in the tree, and if a node $g \in \mathcal I_1$ and $g^\prime$ is a parent of $g$, then $g^\prime \in \mathcal I_1$ as well. The set $\mathcal{I}_0 = \mathcal G \setminus \mathcal{I}_1$ is defined as the complement of $\mathcal{I}_1$ (i.e., $\mathcal{I}_0 \cap \mathcal{I}_1 = \mathcal{G}$). An implication of this construction is described in Assumption~\ref{as:coherence}, formalizing a logical relationship between the node-specific hypotheses.
\begin{assumption}\label{as:coherence}
For all nodes $g \in \mathcal{G},$  $H_{0g}$ holds if and only if  $\cap_{\ell \in \mathcal L_g} H_{0\ell}$ holds. In other words, a null hypothesis holds if and only if all the descendants' null hypotheses also hold.
\end{assumption}
In the context of comparative safety, Assumption~\ref{as:coherence} implies that if we assume that the index drug does not increase the risk of a certain outcome node (e.g., {\em Systolic congestive heart failure}), there will also be no increase in risk in all finer-scale nodes that compose this outcome (e.g., {\em Acute systolic congestive heart failure}, {\em Acute on chronic systolic congestive heart failure}, and others). Note that Assumption \ref{as:coherence} refers to the true unknown null hypotheses rather than the actual decisions based on a TBSS. 
\begin{prop}\label{prop:strong_FWER}
When  Assumption \ref{as:coherence} holds,
 the decisions $\{\indicator\{{\tt pv}_{g} \leq \alpha\}\}_{g \in \mathcal G}$ to reject node-specific null hypotheses $\{H_{0g}\}_{g \in \mathcal{G}}$ controls the FWER for all sets $\mathcal I_0 \subseteq \mathcal G$.
\end{prop}
An important consequence of Proposition~\ref{prop:strong_FWER} is that when using node-specific hypotheses, the probability of finding a statistical alert (node with $ \texttt{pv}_g \leq \alpha$) that is unrelated to the signal is lower than $\alpha.$ Therefore, the practice of selecting areas of future investigation by inspecting all the nodes with small p-values retains multiplicity control if the signal propagates as in Assumption~\ref{as:coherence}.


\section{Tree-based scan statistics with time-to-event data}\label{sec:method}
We consider settings where time-to-event outcomes are measured and available for the analysis. In particular, for each leaf $\ell \in \mathcal L$ of the tree we are given data for $n_\ell$ individuals. As for other TBSSs, we assume that data for the individuals in different leaves are distinct and can be considered realizations of independent random variables. We indicate with $A^{(\ell)}_i\in \{1,0\}$ if individual $i$ in leaf $\ell$ is part of the exposure  ($A^{(\ell)}_i = 1$) or comparator ($A^{(\ell)}_i = 0$) group, and with  $Y^{(\ell)}_i \geq 0$ the time-to-event for leaf (outcome) $\ell \in \mathcal L$. In survival studies, data are typically right-censored, and rather than the actual time when an event occurs, we observe times $U^{(\ell)}_i = \min\{Y^{(\ell)}_i, C^{(\ell)}_i\}$ and event indicators $\delta^{(\ell)}_i \in \{0,1\}$, where $C^{(\ell)}_i$ is the censoring time, $\delta^{(\ell)}_i =1$ for observed events, and $\delta^{(\ell)}_i = 0$ for censored events. As in many survival studies, we assume that censoring is non-informative~\citep[][Sections 3.2]{kalbfleisch:2002}. This allows one to conduct likelihood-based inference by defining a likelihood where uncensored individuals contribute to the likelihood with their probability density functions, while censored individuals contribute to the likelihood with their survival functions (e.g., Equation~\eqref{eq:processLik}). We use $D^{(\ell)}_{\text{full}} = \{U^{(\ell)}_i,\delta^{(\ell)}_i,A^{(\ell)}_i\}_{i=1}^{n_\ell}$ to indicate the \textit{full} data for all the leaves $\ell \in \mathcal{L},$ as opposed to the input data $D^{(\ell)}$ that only contains appropriate data summaries. Full data for a node  $g\in \mathcal{G},$ $D^{(g)}_{\text{full}} = \{U^{(g)}_i,\delta^{(g)}_i, A^{(g)}_i\}_{i=1}^{n_g},$ are obtained by stacking individual-level data for the subjects in leaves in $\mathcal{L}_g$ that are connected to node $g,$ with $n_g = \sum_{\ell \in \mathcal L_g} n_\ell.$

Most survival models can be written in a unified form using point process notation~\citep[e.g.,][Section 1.7]{kalbfleisch:2002}. We define $N^{(g)}_i(t) = \indicator\{U^{(g)}_i \leq t, \delta^{(g)}_i = 1\},$  the counting process registering the number of observed events (0 or 1) until time $t$ for any individual $i \in \{1,\ldots,n_g\}$ in node $g \in \mathcal G.$ The intensity process of $N^{(g)}_i(t)$ is given by
\begin{equation}
\lambda^{(g)}_i(t) = 
\begin{cases}
    \alpha^{(g)}_{0}(t) \indicator\{U^{(g)}_i \geq t\}, \mbox{ if } A^{(g)}_i = 0,\\
    \alpha^{(g)}_{1}(t) \indicator\{U^{(g)}_i \geq t\}, 
     \mbox{ if } A^{(g)}_i = 1,
\end{cases}
\label{eq:intensity}
\end{equation}
where  $\alpha^{(g)}_{a}(t)$ are the hazard rates at time $t$ for an individual in exposure group $a \in \{0,1\}$ in node g $\in \mathcal G.$  We denote with $\tau$ the  maximum length of follow-up for the study such that $U^{(g)}_i \leq \tau$ for $i=1,\ldots,n_g$ and $g \in \mathcal G.$ For a general counting process with intensity process given by Equation~\eqref{eq:intensity}, the likelihood is 
\begin{equation}
\prod_{i=1}^{n_g}
\prod_{0 \leq t \leq \tau }[\lambda^{(g)}_i(t)]^{\Delta N^{(g)}_i(t)} \exp\left\{ - \int_{0}^\tau \sum_{i=1}^{n_g} \lambda^{(g)}_i(t)dt\right\},
\label{eq:processLik}
\end{equation}
where ${\Delta N^{(g)}_i(t)} =N^{(g)}_i(t) - N^{(g)}_i(t-) $ is the increment of the counting process at time $t.$ Intensities  $\lambda^{(g)}_i(t)$ can be defined using a variety of functional forms for the hazard rates $\{\alpha^{(g)}_{a}(t)\}.$  Our proposed survival TBSSs---presented in the next sections---consider the node-specific null hypotheses $H_{0g}: \{\alpha^{(g)}_1(t)  = \alpha^{(g)}_0(t)\}$ for $g \in \mathcal G.$  They differ in their specification for $\{\alpha^{(g)}_{a}(t)\}.$, which leads to different approaches to inference.

\subsection{The Cox TBSS} \label{sec:cox}
In this section, we consider a semi-parametric model for each node of the tree (\emph{coxTBSS}). This method requires access to full patient-level data. For each node $g \in \mathcal{G}$ in the tree, \emph{CoxTBSS} estimates a proportional hazard model~\citep{Cox:1972}, and assumes
\begin{equation}
\alpha^{(g)}_i(t) =  
\kappa^{(g)}_0 (t) \times [ \HR_g]^{A^{(g)}_i},\; \mbox{ for } i = 1 ,\ldots, n_g, \; \mbox{and } g \in \mathcal{G},
\label{eq:ph}
\end{equation}
where $\kappa^{(g)}_0 (t)$ is an arbitrary non-specified baseline hazard function, while $\HR_g$ is the hazard ratio comparing exposed to non-exposed individuals. If the exposure for a node $g \in \mathcal G$ does not affect the hazard, then the parameter $\HR_g=1.$ Therefore, we can consider the null hypothesis $H^{\mbox{\tiny  cox}}_{0g} :\{\HR_g =1 \}$ against the alternatives $H^{\mbox{\tiny  cox}}_{1g} :\{\HR_g \neq 1 \}$ for $g\in \mathcal{G}$, or the corresponding directional hyphotheses. Note that for non-terminal nodes $g$, data from multiple leaves are aggregated, and an event can occur in any leaf connected to $g$, where $\ell \in \mathcal L_g$. For each level of the tree, we estimate fewer models with a broader definition of the outcomes and a larger sample size per model.

In models of the form of Equation~\eqref{eq:ph}, the parameter $\HR_g$ for each node $g \in \mathcal{G}$ can be estimated without specifying a functional form for the baseline hazard $\kappa^{(g)}_0 (t)$ by maximizing the log partial likelihoods
\begin{equation}
\logL^{\mbox{\tiny  cox}}_g (\HR_g) = 
\sum_{i: \delta^{(g)}_i=1}
{A^{(g)}_i}
\log\{
\HR_g\} -
\sum_{i: \delta^{(g)}_i=1}
\log\left\{ \sum_{\{j: U^{(g)}_j \geq U^{(g)}_i\}} \HR^{A^{(g)}_j}_g\right\},\mbox { for }
g\in \mathcal{G},
\label{eq:ph_ll}
\end{equation}
and maximum likelihood estimate $\widehat \HR_g = \arg\max_{\HR_g > 0} \logL^{\mbox{\tiny  cox}}_g (\HR_g).$ Ignoring the tree structure, $H_{0g}$ can be tested using the (partial) log-likelihood ratio test 
\begin{equation}
T^{\mbox{\tiny  cox}}_g = 
2\times\{
\logL^{\mbox{\tiny  cox}}_g (\widehat \HR_g)
- 
\logL^{\mbox{\tiny  cox}}_g (1)
\}, 
\label{eq:ph_lrt}
\end{equation}
which (if considering a single node) is asymptotically distributed as a chi-square random variable with one degree of freedom under the null distribution~\citep{cox:1975}.

To derive node-specific p-values and account for the tree structure of the outcome, we avoid making assumptions about the node-specific baseline hazard function and propose a permutation scheme (\algname~\ref{alg:coxTbss}).
A standard permutation algorithm assumes exchangeability among individuals and considers the global null hypothesis that the distributions of adverse events between the two exposure groups are equal, implying $H^{\mbox{\tiny  cox}}_{0g} =\{\HR_g =1 \}$ for all $g \in \mathcal G.$ In non-randomized studies (such as the ones considered here), individuals are rarely exchangeable. However, they might be approximately exchangeable given a set of covariates.  In a typical TBSS analysis, important covariates are used to compute a propensity score, which is then used to adjust for confounding. If a matching strategy is used (see \secname~\ref{sec:simulations}), permutations can be stratified between matched sets. More broadly, the quantiles of the propensity score can be used to define strata of approximately exchangeable individuals, and permutations can occur within these strata (see \secname~\ref{sec:application}).
\begin{algorithm}
\setstretch{0.9}
\caption{Permutation algorithm for \emph{coxTBSS}}\label{alg:coxTbss}
\SetKwInOut{Input}{Input}\SetKwInOut{Output}{Output} \SetKwInOut{Procedure}{Procedure}

\Input{}
\BlankLine
Patient-level data including:\\
\hspace{5pt} 1) 
leaf code for the outcome $\ell \in \mathcal L$\\
\hspace{5pt} 2) 
time to event $\{U^{(\ell)}_i\}$\\
\hspace{5pt} 3) 
censoring status $\{\delta^{(\ell)}_i\}$\\
\hspace{5pt} 4) exposure indicator, $\{A^{(\ell)}_i\}$.\\
Node list $\mathcal G$ \\ 
Map of each node to the corresponding leaves $\mathcal I_g$ for $g \in \mathcal G$\\
Number of Monte Carlo replicates $R$

\Output{}
\BlankLine
Test statistic $T$ \\
node-specific test statistics $\{T_g\}$ for $g \in \mathcal G$ \\
node-specific p-values $\{\texttt{pv}_g\}$ for $g \in \mathcal G$

\BlankLine \BlankLine
\Procedure{} \BlankLine

\For{$r = 1,\ldots,R$ }{
\vspace{0.1cm}
\noindent\textbf{1.} for each leaf $\ell \in \mathcal L$ generate data under the null hypothesis by permuting the exposure indicator of the original data. Permutations should be stratified only to permute approximately exchangeable individuals, such as matched pairs.\\
\vspace{0.2cm}
    \noindent\textbf{2.}  
    With the data in 1., compute the hazard ratio maximizing Equation~\eqref{eq:ph_ll}
    and the statistic $T^{(r)}_g$ of Equation~\eqref{eq:ph_lrt} for each node in $\mathcal G,$ and let $T^{(r)} = \max_{g \in \mathcal G} \{T^{(r)}_{g}\}.$
}
\noindent\textbf{3.} Compute the p-values for each node as follows:\\
    \For{$g \in \mathcal {G}$}{
     $$
      {\tt pv}_{g} = \frac{1+\sum_{r=1}^{R} \indicator\{ T^{(r)} \geq   T_{g} \}}{R+1}.
     $$
}
\end{algorithm}
%
%
\subsection{The exponential survival TBSS}\label{sec:exp}
In this section, we describe a parametric TBSS to analyze survival data (\emph{exponentialTBSS}). The \emph{exponentialTBSS} algorithm assumes exponentially distributed time-to-event data for the events defined at the leaves of the tree, implying a constant hazard rate for each node of the tree.  Unlike \emph{coxTBSS} (\secname~\ref{sec:cox}), \emph{exponentialTBSS} requires only a few data summaries for inference, namely person time and number of events for each leaf of the tree.

Using the counting process notation, we assume that the intensity functions are constant over time, meaning $\alpha^{(g)}_a (t) = \alpha^{(g)}_a$ for $a \in \{0,1\}$ and $t \in [0,\tau]$. This corresponds to letting (uncensored) time-to-events $Y^{(\ell)}_i | A^{(\ell)}_i =a$ be distributed according to an exponential random variable with mean $[a^{(\ell)}_a]^{-1}$, where $a \in \{0,1\}$ and $\ell \in \mathcal{L}$. Under the assumption that the censoring mechanism is non-informative, the log-likelihood for the parameters $\boldsymbol \alpha^{(\ell)} = (\alpha^{(\ell)}_0,\alpha^{(\ell)}_1)$ for $\ell \in \mathcal L$ is given by
\begin{equation}
\logL^{\mbox{\tiny exp}} ( \boldsymbol \alpha^{(\ell)}) =
\left\{
\sum_{a=0,1}\log(\alpha^{(\ell)}_a) \sum_{i: A_i =a} {\delta^{(\ell)}_i}
   -\sum_{a=0,1}\alpha^{(\ell)}_a \sum_{i: A_i =a} U^{(\ell)}_i
   \right\}.
\label{eq:exp_lik}
\end{equation}
The sufficient statistics for the inference on $\boldsymbol \alpha^{(\ell)}$ are given by $S^{(\ell)}_a = (R^{(\ell)}_a, W^{(\ell)}_a),$ where $R^{(\ell)}_a = \sum_{i=1}^{n_\ell} \delta^{(\ell)}_i \indicator\{A^{(\ell)}_i=a\}$ is the total number of observed events at the end of the study, and $W^{(\ell)}_a = \sum_{i=1}^{n_\ell} U^{(\ell)}_i \indicator\{A^{(\ell)}_i=a\}$ is the total survival (or person) time for exposed ($a=1$) and unexposed ($a=0$) individuals. The maximum likelihood estimates (MLEs) of the HRs for each leaf are $\widehat \alpha^{(\ell)}_a = R^{(\ell)}_a/ W^{(\ell)}_a$, where $a \in \{0,1\}$. Sufficient statistics can be computed for each node of the tree by summing up the statistics for the corresponding leaf nodes, in particular $R^{(g)}_a = \sum_{\ell \in \mathcal{L}_g}R^{(\ell)}_a$ and $W^{(g)}_a = \sum_{\ell \in \mathcal{L}_g} W^{(\ell)}_a$. These statistics can in turn be used to define node-specific estimates of  $\widehat \alpha^{(g)}_a = R^{(g)}_a/W^{(g)}_a$ for $a \in \{0,1\}$ and $g \in \mathcal{G}$.

For each node $g \in \mathcal G$, we test the null hypothesis $H^{\mbox{\tiny exp}}_{0g}: \{ \alpha^{(g)}_1 = \alpha^{(g)}_0 \}$ against the alternative $H^{\mbox{\tiny exp}}_{1g}: \{ \alpha^{(g)}_1 \neq \alpha^{(g)}_0 \}$ or corresponding directional hypotheses. As test statistics, we can use the LRT, defined for the leaf-level data (Equation~\eqref{eq:exp_lik}), and use the same functional form for the node-specific test statistics
\begin{equation}
T^{\mbox{\tiny exp}}_g =  2 \left\{\sum_{a=0,1} 
R^{(g)}_a  \times \log\left( \frac{R^{(g)}_a}{W^{(g)}_a} \right) - 
(R^{(g)}_0 + R^{(g)}_1) \times \log\left( \frac{R^{(g)}_0 + R^{(g)}_1}{W^{(g)}_0 + W^{(g)}_1} \right) 
\right\},
\label{eq:lrt_exp}
\end{equation}
for $g \in \mathcal{G},$ and their maximum as a test statistic for the global null hypothesis.
The next step is to propose an algorithm to  approximate the distributions of Equation~\eqref{eq:lrt_exp}
under the global null $H_0$ of no effect in any node of the tree. In the considered exponential model, under the global null, time-to-events are distributed as $Y^{(\ell)}_i | A^{(\ell)}_i = a \overset{H_0}{\sim} \text{Exp}(\alpha^{(\ell)}_{H_0})$ for $i \in \{1,\ldots,n_\ell \}$ and $a \in \{0,1\}$, where the parameters $\alpha^{(\ell)}_{H_0}$ can be consistently estimated via  $\widehat \alpha^{(\ell)}_{H_0}  =(R^{(\ell)}_0 + R^{(\ell)}_1)/(W^{(\ell)}_0 + W^{(\ell)}_1)$ for $\ell \in \mathcal L.$

For inference with TBSSs, we are not interested in individuals' time-to-events or censoring times but in the distribution of the statistics in Equation~\eqref{eq:lrt_exp} and their maximum under the global null. Samples from this distribution can be generated using the fact that, under $H_0$, the counting processes for the exposure $a\in \{0,1\}$ corresponding to the log-likelihood of Equation~\eqref{eq:exp_lik} are two independent homogeneous Poisson processes with the same intensity $\alpha^{(\ell)}_{H_0}$. This implies that, conditioning on the total person times $W^{(\ell)}_a$ for $a \in \{0,1\}$ and $\ell \in \mathcal L$, the number of events
$R^{(\ell)}_a \overset{H_0}{\sim} \mbox{Poisson}(\alpha^{(\ell)}_{H_0} \times W^{(\ell)}_a)$ for $a \in \{0,1\}$. The algorithm to perform \emph{exponentialTBSS} inference can be found in \algname~\ref{alg:exp_surv}.
\begin{algorithm}
\setstretch{0.9}
\caption{Bootstrap algorithm for \emph{exponentialTBSS}}\label{alg:exp_surv}
\SetKwInOut{Input}{Input}\SetKwInOut{Output}{Output} \SetKwInOut{Procedure}{Procedure}

\Input{}
\BlankLine
Leaf-level data summaries for exposure groups $a=0,1$ and $\ell \in \mathcal L$ \\
\hspace{5pt} 1) total person times  $W^{(\ell)}_a = \sum_{i=1}^{n_\ell} U^{(\ell)}_i \indicator\{A^{(\ell)}_i=a\}$
\\
\hspace{5pt} 2) total number of observed events  $R^{(\ell)}_a = \sum_{i=1}^{n_\ell} \delta^{(\ell)}_i \indicator\{A^{(\ell)}_i=a\}.$
\\
Node list $\mathcal G$ \\ 
Map of each node to the corresponding leaves $\mathcal I_g$ for $g \in \mathcal G$\\
Number of Monte Carlo replicates $R$

\Output{}
\BlankLine
Test statistic $T$ \\
node-specific test statistics $\{T_g\}$ for $g \in \mathcal G$ \\
node-specific p-values $\{\texttt{pv}_g\}$ for $g \in \mathcal G$

\BlankLine \BlankLine
\Procedure{} \BlankLine

\For{$r = 1,\ldots,R$ }{

\noindent\textbf{1.} for each leaf $\ell \in \mathcal L$ generate data under the null hypothesis:\\
       \For {$\ell \in \mathcal {L}$}{
    $$
    R^{(r,\ell)}_a \sim \mbox{Poisson}\left(
    \widehat \alpha^{(\ell)}_{H_0} \times W^{(\ell)}_a\right), \mbox{ for } a=0,1.
    $$  
    with  $\widehat \alpha^{(\ell)}_{H_0}  =(R^{(\ell)}_0 + R^{(\ell)}_1)/(W^{(\ell)}_0 + W^{(\ell)}_1)$ for $\ell \in \mathcal L.$ 
}
    \noindent\textbf{2.}  Compute the statistic $T_g$ for each node in $\mathcal G$:\\
     \For{$g \in \mathcal {G}$}{
     \begin{align*}
     R^{(r,g)}_a &= \sum_{\ell \in \mathcal{L}_g}  R^{(r,\ell)}_a,\\
     T^{(r)}_g &= 2 \left\{ \sum_{a=0,1} 
R^{(r,g)}_a  \times \log\left( \frac{R^{(r,g)}_a}{W^{(g)}_a} \right)
-(R^{(r,g)}_0 + R^{(r,g)}_1) \times \log\left( \frac{R^{(r,g)}_0 + R^{(r,g)}_1}{W^{(g)}_0 + W^{(g)}_1} \right)
\right\}.
    \end{align*}
      }
    Let $T^{(r)} = \max_{g \in \mathcal G} \{T^{(r)}_{g}\}$.
}
\noindent\textbf{3.} Compute the p-values for each node:\\
    \For{$g \in \mathcal {G}$}{
     $$
      {\tt pv}_{g} = \frac{1+\sum_{r=1}^{R} \indicator\{ T^{(r)} \geq   T_{g} \}}{R+1}.
     $$
}
\end{algorithm}
Although modeling assumptions are different, \algname~\ref{alg:exp_surv} shares some similarity to some versions of the unconditional Poisson TBSS, where the `expected counts'---the counts for the exposure group---are scaled to account for differential person time in the exposure groups~\citep[e.g.][]{wang:2018,huybrechts:2021,suarez:2023}. These applications of unconditional Poisson TBSS assume that the expected counts are known based on population event rates. However, in practice they are often estimated from a comparator group with potentially limited sample size. In contrast,  \textit{exponentialTBSS} considers the events for both exposure groups as unknown and appropriately accounts for uncertainty in their estimates.  Failure to account for this uncertainty---i.e., using a `traditional' Poisson TBSS---might lead to inflated type-I error rates and/or loss of power. This issue with `traditional' Poisson TBSS has been highlighted, for example, in~\citet{thuy:2024}.

\subsection{Robust exponential TBSS}\label{sec:robust}
The method presented in \secname~\ref{sec:exp} assumes that hazard rates for all nodes in the tree are constant over time or, equivalently, that time-to-events are exponentially distributed. If patient-level data are available, we can test this assumption and potentially consider alternative parametric distributions for the leaf-level time-to-events (e.g., log-normal or Weibull distributions). 

Instead of implementing a plethora of parametric survival TBSSs, we propose a robust inference procedure that uses the \emph{exponentialTBSS} as a working model rather than a correctly specified model. We call this method \emph{robustTBSS}. Our proposed model leverages robust inference results for survival analysis presented, for example, in \citet{hjort:1992}.

When we erroneously assume that the hazard rates $\{\alpha^{(g)}_a(t), a=0,1\}$ for a certain node $g \in \mathcal{G}$ and times $t \in [0,\tau]$ are constant over time, the estimates $\widehat \alpha^{(g)}_a = R^{(g)}_a/W^{(g)}_a$ converge in probability to
$
\dot \alpha^{(g)}_a = 
{\int_{0}^\tau \alpha^{(g)}_a(s) y^{(g)}(s) ds}
/{\int_{0}^\tau y^{(g)}(s)ds}
$
as $n_g \to \infty,$ where $\lim_{n_g \to \infty } \sum_{i=1}^{n_g} \indicator\{A^{(g)}_i = a\} = \infty$ for $a=0,1$ and $y^{(g)}(s) = \pr(U^{(g)}_1 \geq s)$ is the limit in probability of  $n^{-1}_g \sum_{i= 1}^{n_g} I\{U_i \geq t\}$. This result can be found in \cite[Example~2.1]{hjort:1992}.

The limits $\dot \alpha^{(g)}_a$ are weighted averages of the true hazard rates and can be used for testing purposes. In fact, rejecting the hypothesis that $\dot \alpha^{(g)}_0 =\dot \alpha^{(g)}_1$ for a node $g \in \mathcal {G}$ is still a well-posed problem, irrespective of the functional form of the hazard rate.  However, if the model is not correctly specified, parametric bootstraps, such as \algname~\ref{alg:exp_surv}, which assume the leaf-level models are correctly specified, do not estimate the variance correctly. Consequently, power might decrease, or the type-I error rate might be inflated.

Even when the exponential survival model is not correctly specified, the following result holds for the leaf estimates of the hazard rates:
\begin{equation}
\sqrt{n_\ell}(\widehat \alpha^{(\ell)}_a - \dot \alpha^{(\ell)}_a) \overset{d}{\to} \mathcal{N}\left(0, \frac{K^{(\ell)}_a}{[J^{(\ell)}_a]^2}\right), \mbox{ for } a=0,1
\label{eq:asymp_rob}
\end{equation}
as $n_{\ell} \to \infty$  and  $\lim_{n_\ell \to \infty } \sum_{i=1}^{n_\ell} \indicator\{A^{(\ell)}_i = a\} = \infty$. Here $\overset{d}{\to}$ indicates convergence in distribution, 
$$
J^{(\ell)}_a = [\dot \alpha^{(\ell)}_a]^{-1}
\int_{0}^{\tau} y^{(\ell)}(s) \alpha^{(\ell)}(s) ds,
$$
and
$$
K^{(\ell)}_a = [\dot \alpha^{(\ell)}_a]^{-1}
\int_{0}^{\tau} y^{(\ell)}(s) \alpha^{(\ell)}(a) ds +
2[\dot \alpha^{(\ell)}_a]^{-2} \int_{0}^{\tau} \int_{0}^t y^{(\ell)}(s)\{ \alpha^{(\ell)}(s) -\dot \alpha^{(\ell)}_a\}ds dt, 
$$
leading to an estimate of the asymptotic variances of $\widehat\alpha^{(\ell)}_a$ given by
$$
\mathbb{V}\mbox{ar}(\widehat \alpha^{(\ell)}_a) = 
\frac{
[\widehat \alpha^{(\ell)}_a]^2}{[R^{(\ell)}_a/n^{(\ell)}_a]^2}
\frac{1}{\widehat \alpha^{(\ell)}_a} \sum_{i : A^{(\ell)} =a} \left(\frac{\delta^{(\ell)}_i}{\widehat \alpha^{(\ell)}_a} - U^{(\ell)}_i\right)^2,
$$
where $n^{(\ell)}_a = \sum_{i=1}^{n_\ell} \indicator\{A^{(\ell)}_i =a\}$ is the number of individuals in exposure group $a \in \{0,1\}$.

Rather than directly using Equation~\eqref{eq:asymp_rob}, we first reparameterize the HRs in Equation~\eqref{eq:exp_lik} with $\beta^{(g)}_0 = -\log\{ \alpha^{(g)}_0\}$ and $\beta^{(g)}_1 = -(\log\{ \alpha^{(g)}_1\} - \log\{ \alpha^{(g)}_0\}).$ These are the coefficients of an exponential survival regression model, where $\beta^{(g)}_0$ are intercepts, and $\beta^{(g)}_1$ are the deviations from the intercept for exposed individuals on the log scale. Using this parameterization to implement \textit{robustTBSS} is convenient. First, it avoids generating a negative average number of events, which could happen if the approximation in Equation~\eqref{eq:asymp_rob} is used with a small sample size. Additionally, it simplifies the generation of data replicates under the global null hypothesis. In fact, the null hypotheses that exposure does not affect the outcome variable for node $g\in \mathcal{G}$ can be written as the hypothesis $H^{\mbox{\tiny rob}}_{0g} = \{\beta^{(g)}_1 =0\}$ against the alternatives $H^{\mbox{\tiny rob}}_{1g} = \{\beta^{(g)}_1 \neq 0\}.$ Directional alternatives could be considered. The algorithm proceeds by first generating the coefficients  $(\beta^{(g)}_0, \beta^{(g)}_1)$ under the global null, and then by mapping this estimate to the number of observed events. Details are described in \algname~\ref{alg:robust}.
\begin{algorithm}
\setstretch{0.9}
\caption{Bootstrap algorithm for \emph{robustTBSS}}\label{alg:robust}
\SetKwInOut{Input}{Input}\SetKwInOut{Output}{Output} \SetKwInOut{Procedure}{Procedure}

\Input{}
\BlankLine
Leaf-level data summaries for exposure groups $a=0,1$ and $\ell \in \mathcal L$ \\
\hspace{5pt} 1) number of individuals $n^{(\ell)}_a,$\\
\hspace{5pt} 2) total person times $W^{(\ell)}_a= \sum_{i=1}^{n_\ell} U^{(\ell)}_i \indicator\{A^{(\ell)}_i=a\},$\\
\hspace{5pt} 3) total number of observed events $R^{(\ell)}_a =  \sum_{i=1}^{n_\ell} \delta^{(\ell)}_i \indicator\{A^{(\ell)}_i=a\},$\\
\hspace{5pt} 4) sum of squares of the person time $W^{(\ell,2)}_a=\sum_{i=1}^{n_\ell} [U^{(\ell)}_i]^2 \indicator\{A^{(\ell)}_i=a\},$\\
\hspace{5pt} 5) person time among the uncensored individuals $O^{(\ell)}_a =\sum_{i=1}^{n_\ell} \delta^{(\ell)}_i U^{(\ell)}_i \indicator\{A^{(\ell)}_i=a\}$\\

Node list $\mathcal G$ \\ 
Map of each node to the corresponding leaves $\mathcal I_g$ for $g \in \mathcal G$\\
Number of Monte Carlo replicates $R$

\Output{}
\BlankLine
Test statistic $T$ \\
node-specific test statistics $\{T_g\}$ for $g \in \mathcal G$ \\
node-specific p-values $\{\texttt{pv}_g\}$ for $g \in \mathcal G$

\BlankLine \BlankLine
\Procedure{} \BlankLine

\For{$r = 1,\ldots,R$ }{

\noindent\textbf{1.} For each leaf $\ell \in \mathcal L$ generate data under the null hypothesis:\\
       \For {$\ell \in \mathcal {L}$}{
$$
(\beta^{(r,\ell)}_0,\beta^{(r,\ell)}_1) \sim \mathcal{N}_2 \left( (\widehat \beta^{(\ell)}_0,0),  J^{(\ell)}K^{(\ell)}J^{(\ell)} \right),
$$
where 
$
J^{(\ell)}_{11} = 
e^{-\widehat{\beta}_0^{(\ell)}}[ W^{(\ell)}_0  +
e^{- \widehat{\beta}_1^{(\ell)}} W^{(\ell)}_1]
,\;\;\;$ 
$J^{(\ell)}_{12} =J^{(\ell)}_{21} = J^{(\ell)}_{22} =   e^{-\widehat{\beta}_0^{(\ell)}} W^{(\ell)}_0,$ \newline
$\;$ 
\newline
and
$
K^{(\ell)}_{11} = 
n^{(\ell)} +
e^{-\widehat{\beta}_0^{(\ell)}}\left\{e^{-\widehat{\beta}_0^{(\ell)}}[
W^{(\ell,2)}_0 +
 W^{(\ell,2)}_1 e^{-2\widehat{\beta}_1^{(\ell)}} 
]
-2
[
O^{(\ell)}_0 +
O^{(\ell)}_1e^{-\widehat{\beta}_1^{(\ell)}}
]\right\},
$
$
K^{(\ell)}_{12} = 
K^{(\ell)}_{21} = 
K^{(\ell)}_{22} = 
n^{(\ell)}_1 +  e^{-\widehat{\beta}_0^{(\ell)} -\widehat{\beta}_1^{(\ell)} }[
 e^{-\widehat{\beta}_0^{(\ell)} -\widehat{\beta}_1^{(\ell)} } W^{(\ell,2)}_1  -2  O^{(\ell)}_1 ].
$

$\;$

Then set     
    $
    R^{(r,\ell)}_0 = \exp(-\beta^{(r,\ell)}_0) \times W^{(\ell)}_0, \;\;
    R^{(r,\ell)}_1 = \exp(-\beta^{(r,\ell)}_0 - \beta^{(r,\ell)}_1) \times W^{(\ell)}_1.
    $  
}
    \noindent\textbf{2.}  Compute the statistic $T_g$ for each node in $\mathcal G$:\\
     \For{$g \in \mathcal {G}$}{
     \begin{align*}
     R^{(r,g)}_a &= \sum_{\ell \in \mathcal{L}_g}  R^{(r,\ell)}_a,\\
     T^{(r)}_g &=  2 \left\{   \sum_{a=0,1} 
R^{(r,g)}_a  \times \log\left( \frac{R^{(r,g)}_a}{W^{(g)}_a} \right)
-(R^{(r,g)}_0 + R^{(r,g)}_1) \times \log\left( \frac{R^{(r,g)}_0 + R^{(r,g)}_1}{W^{(g)}_0 + W^{(g)}_1} \right)
\right\}.
    \end{align*}
      }
       Let $T^{(r)} = \max_{g \in \mathcal G} \{T^{(r)}_{g}\}.$
}
\noindent\textbf{3.} Compute the p-values for each node:\\
    \For{$g \in \mathcal {G}$}{
     $$
      {\tt pv}_{g} = \frac{1+\sum_{r=1}^{R} \indicator\{ T^{(r)} \geq   T_{g} \}}{R+1}.
     $$
}
\end{algorithm}


\section{Simulation study}\label{sec:simulations}
To investigate the performance of the proposed survival TBSSs, we considered a comparative safety study with two drugs: an index ($a=1$) and a comparator ($a=0$) drug. To mimic a realistic study, we consider the following steps: We generate a cohort of patients, fit a propensity score (PS) model, perform one-to-one PS matching, and then apply TBSS analysis to the matched data~\citep[e.g.,][]{wang:2018}.
\subsection{Generating simulated data}
For the patient population, we refer to patients’ characteristics with $\textbf{X}^{(\ell)}_i$ for characteristics that are available for the analysis, and $\textbf{Z}^{(\ell)}_i$ for unmeasured characteristics that can act as confounding factors, where $i=1,\ldots,n_\ell$ and $\ell \in \mathcal L.$

We generate patient-level data using, for each leaf $\ell \in \mathcal{L}$ and exposure group $a=0,1$,  the following steps:
\begin{enumerate}
    \item Sample patient characteristic  $(\textbf{X}^{(\ell)}_i,\textbf{Z}^{(\ell)}_i)\sim F_{X,Z}$ from a distribution $F_{X,Z},$ for $i=1,\ldots,n_\ell$ and $\ell \in \mathcal L.$
    
    \item Assign each patient to treatment $A^{(\ell)}_i \in \{0,1\}$ by sampling the indicator from a binomial distribution with probability $[1 + \exp\{-(\eta_1 +[\textbf {X}^{(\ell)}_i]^\intercal  \boldsymbol \zeta_1   +  [\mathbf Z^{(\ell)}_i]^{\intercal} \boldsymbol \gamma_1)\}]^{-1}.$ Here, the parameter $\eta_1$ controls the proportion of treated/untreated in the population, while $\boldsymbol \zeta_1$ and $\boldsymbol \gamma_1$ express the effect of measured and unmeasured confounding, respectively, on the treatment allocation mechanism. 
   
    \item Sample time-to-event data $Y^{(\ell)}_i$ for  $i=1,\ldots, n_\ell$ from a Weibull distribution with  shape parameter $\kappa^{(\ell)}_a$ and scale $\lambda^{(\ell)}_a = \lambda_0/\exp\{[\mathbf{X}^{(\ell)}_i]^{\intercal} \boldsymbol \zeta_0  + [\mathbf Z^{(\ell)}_i]^{\intercal} \boldsymbol \gamma_0 + \omega^{(\ell)} A^{(\ell)}_i\}.$ The parameter $\omega^{(\ell)}$ induces an ``exposure effect'' for the leaf $\ell \in \mathcal L$,  while $\boldsymbol \zeta_0$ and $\boldsymbol \gamma_0$ express the effect of measured and unmeasured confounding, respectively, on the time-to-event. In \secname~\ref{sec:composite} of the supplementary material, we also discuss a different strategy for inducing the exposure effect at higher-level nodes.
    
    \item Sample time-to-censoring data $C^{(\ell)}_i$ for  $i=1,\ldots, n_\ell$ from a Weibull distribution with  shape parameter $\kappa^{(\ell)}_a$ (as in step 3) and scale parameters $\gamma^{(\ell)}_a = ([\lambda^{(\ell)}_a]^{\kappa^{(\ell)}_a} p_a/(1-p_a))^{1/\kappa^{(\ell)}_a}.$ This choice leads to an average proportion of censored event $p_a \in (0,1).$
    
    \item Set the observed time $U^{(\ell)}_i = \min(Y^{(\ell)}_i, C^{(\ell)}_i)$ and the censoring indicator $\delta^{(\ell)}_i = \indicator\{Y^{(\ell)}_i \leq C^{(\ell)}_i\}.$
\end{enumerate}

By changing the simulation parameters $(\kappa^{(\ell)}_a,  \eta_1, \boldsymbol  \zeta_1,\boldsymbol  \gamma_1, \lambda_0, \boldsymbol{\zeta}_0, \boldsymbol  \gamma_0, \omega^{(\ell)},p_a)$, we can create broadly different scenarios with different distributions of time to event, number of observed events, and differences in distribution between the exposure groups. Parameters for all the simulation scenarios considered in this section are reported in \tablename~\ref{tab:mode_and_deriv} of the supplementary material together with other implementation details. Before analyzing the data with various TBSSs, we pre-process the generated data to adjust for measured confounding via a one-to-one PS matching; see supplementary material \secname~\ref{sec:ps_matching} for the details.

\subsection{Simulation details and metrics.}
We consider five TBSS methods to analyze the matched data: 
\begin{enumerate}
\item \textit{bernoulliTBSS:} After PS-matching for each matched set, only the events that have occurred from the index date to the shortest follow-up time within the matched set are considered for analysis. The others are discarded. This ensures that exposure groups are comparable, even if the follow-up times differ between the groups. Refer to \citet{russo:2024} for a discussion on this point. 
\item \textit{poissonTBSS:} This is commonly employed for comparative studies in pharmacology~\citep[e.g.,][]{huybrechts:2021}. This version of TBSS tests if, for each node in the tree, the number of observed events differs from the number of events for a comparator group (appropriately scaled).   Unlike Bernoulli TBSS, no event is discarded because of differential follow-up, but the uncertainty in the expected number of events---estimated as the number of events in the comparator group---is ignored. 
\item \textit{coxTBSS}: The method described in  \secname~\ref{sec:cox}.
\item \textit{exponentialTBSS}: The method described in  \secname~\ref{sec:exp}.
\item \textit{robustTBSS}: The method described in  \secname~\ref{sec:robust}.
\end{enumerate}

For all the TBSS methods under consideration, we use a two-tailed alternative and a $5\%$ threshold to test the node-specific p-values. We compare the methods with three metrics. First, we consider the probability of rejecting the global null $H_0$, i.e., the probability of the event  $\{\pr(T > T^{\text{\small obs}}| H_0) \leq 0.05\}$. We refer to this metric as \textit{global power} (\figurename~\ref{fig:sim1}). In addition, we consider the proportion of \textit{false positives} (\figurename~\ref{fig:fdr}) and \textit{true positives} (\figurename~\ref{fig:tpr}), defined as the number of node-specific null hypotheses in the sets $\mathcal{I}_0$ and $\mathcal{I}_1$ (see \secname~\ref{sec:method}) divided by the total number of rejections. We also consider the runtime of the five methods (\figurename~\ref{fig:sim_time}).

In our simulations, we use three distributional assumptions for the survival time of the patient-level data. We assume time-to-events follow an exponential or Weibull distribution. As a third assumption, we use exponentially distributed time-to-events, but the exposure effect is delayed; hence, the proportional hazards assumption does not hold. To inject signals in the tree, we consider three settings. The first injects a signal only in one leaf node with a correctly specified PS model; the second considers the same signals but a mispecified PS-score model due to unmeasured confounders; the third injects signals in a higher-level node rather than a leaf (\secname~\ref{sec:composite}). Refer to supplementary \figurename~\ref{fig:sim_tree} and its caption for a more detailed description of the signal structure and a graphical representation of the tree used in the simulations.

For each of the nine combinations of signal and time-to-event distributions, we consider $20$ different values for `exposure effect' ($\omega^{(25)} \in \{0,\ldots,\log(2.5)\}$). For each of these $180$ scenarios, we generate $1,000$ datasets with $7800$ individuals,  $600$ for each of $13$ leaves. \tablename~\ref{tab:mode_and_deriv} in the supplementary material provides additional details.
\begin{figure}[ht!]
    \centering
    \includegraphics[width =\textwidth]{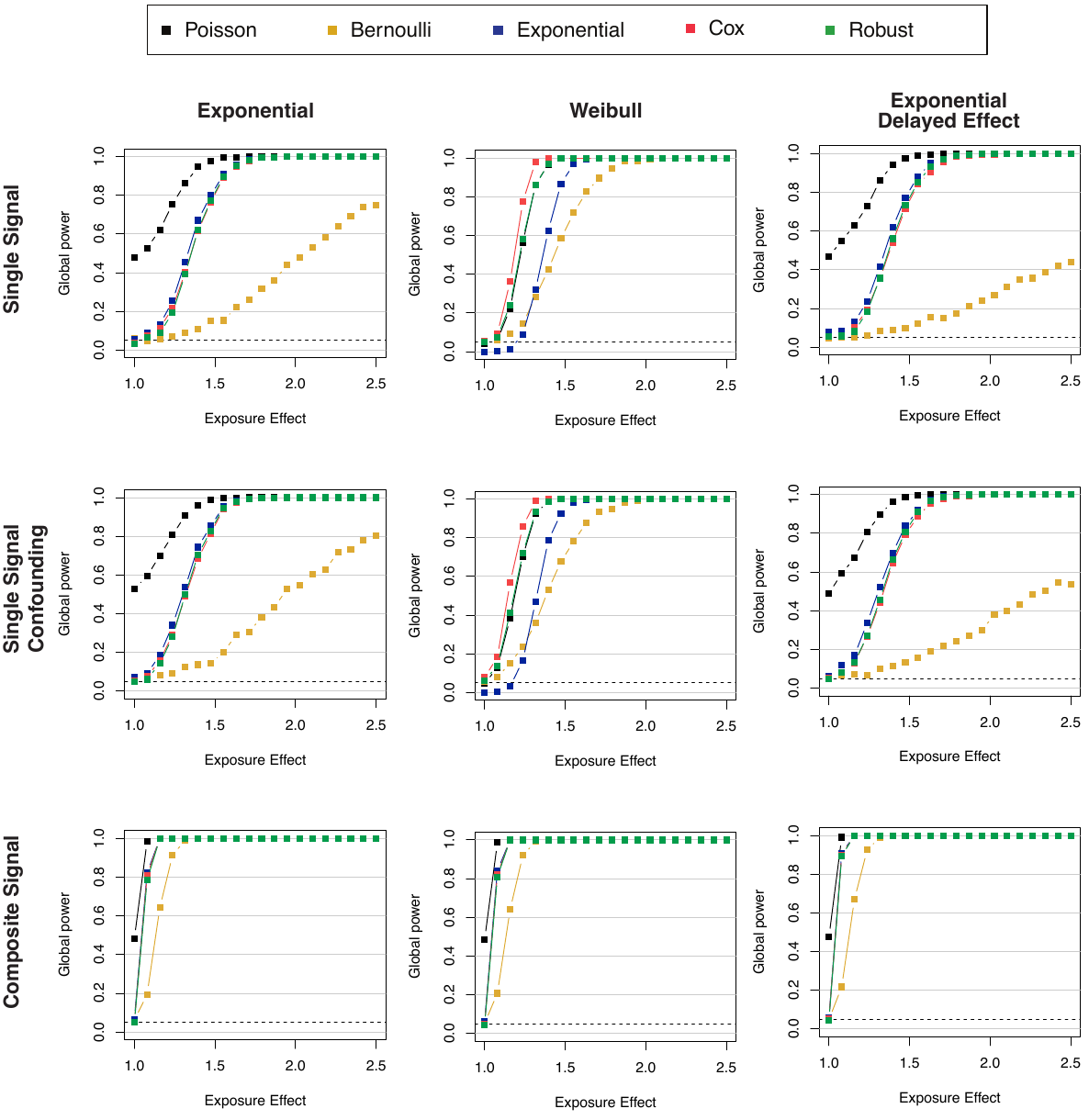}
    \caption{Proportion of times the global null hypothesis $H_0$ was rejected (global power) for the simulations described in \secname~\ref{sec:simulations}. Each column represents a distributional assumption for the time-to-event data: column 1 shows exponential time-to-event data, column 2 shows Weibull time-to-event data, and column 3 shows exponential time-to-event data with a delayed effect. The three rows are the three different simulation scenarios. In the first row, the treatment effect is added to a single leaf (Node-25 of the tree in \figurename~\ref{fig:sim_tree}), and the propensity score is correctly specified. In the second row, the propensity score model is mispecified because of unmeasured confounders, and in the third row, the exposure effect is included at a higher level (Node-6  in \figurename~\ref{fig:sim_tree}). The x-axis represents exposure effects $\exp\{\omega^{(25)}\}$ or the effect induced on Node-6 (third row). An exposure effect = 1 indicates no effect and corresponds to the type-I error of the global null hypothesis.
    }
    \label{fig:sim1}
\end{figure}

\subsection{Simulation Results}
\figurename~\ref{fig:sim1} shows the proportion of times the global null $H_0$ was rejected (global power) for different values of the exposure effect and data-generating mechanisms.  In the case of exponentially distributed time-to-event data (row 1, column 1), \textit{poissonTBSS} is severely anti-conservative, with a probability of rejecting the global null $H_0$  of approximately $48\%$. The other four approaches have a type-I error rate close to the nominal level of $5\%$.  The  \textit{exponentialTBSS}, \textit{robustTBSS}, and  \textit{coxTBSS} methods share similar performance also in terms of global power, with superimposed power curves for most values of the exposure effect. Off note, \textit{exponentialTBSS} is correctly specified in this scenario. The \textit{bernoulliTBSS}  has substantially lower global power than the other methods,  approximately $75\%$ for $\omega^{(25)} = \log(2.5)$---the largest considered signal---compared to $100\%$ for the other four methods. The poor performance of \textit{bernoulliTBSS} is common to all the considered simulations. This is likely due to the number of events that are discarded to ensure comparability in follow-up time between the two exposure groups, which, given the relatively small sample size considered in these simulations, considerably affects performance.

For Weibull distributed time-to-event data (row 1, column 2), all five methods achieve a type-I error lower than $5\%.$ The \textit{exponentialTBSS} algorithm---here misspecified---is conservative with a type-I error close to $0$ when there is no exposure effect and power considerably lower than  \textit{poissonTBSS}, \textit{coxTBSS}, and \textit{robustTBSS}. The \textit{coxTBSS} algorithm has higher power than \textit{robustTBSS}. For example, for $\omega^{(25)} = \log(1.30)$, \textit{coxTBSS} has a power of approximately $98\%$ compared to $86\%$ of \textit{robustTBSS}.

With delayed effect (row 1 column 3), \textit{poissonTBSS} is anti-conservative, with a probability of rejecting the global null when it is true of approximately $47\%$ compared to the nominal $5\%$. \textit{exponentialTBSS} has a type-I error of approximately $8\%$. \textit{coxTBSS}  and \textit{robustTBSS} have type-I error close to the nominal level, approximately $5.6\%$, while \textit{bernoulliTBSS} has a type-I error of approximately $ 5\%$. Overall, \textit{exponentialTBSS}, \textit{coxTBSS}, and  \textit{robustTBSS} have similar power curves, while \textit{poissonTBSS} is anti-conservative and  \textit{bernoulliTBSS} has lower power than the other approaches.

The second row in \figurename~\ref{fig:sim1} shows settings with a moderate amount of mispecification added to the PS model. All methods are relatively robust to moderate mispecification, with only a slight increase in false positive results. For example, with exponentially distributed data, \textit{poissonTBSS} has a type-I error rate of approximately $53\%$, compared to $48\%$ in the absence of unmeasured confounding. However, the relative comparisons among the methods remained similar to those in scenarios without unmeasured confounding, suggesting that mispecification of the PS model or the presence of unmeasured confounders affects all methods similarly.

The third row of \figurename~\ref{fig:sim1} considers scenarios with composite signals. Also in this case, \textit{poissonTBSS} has a large type-I error, approximately $47\%$ compared to the nominal $5\%.$ The remaining four methods closely control type-I error and have a similar power curve for increasing levels of the exposure effect parameters. \textit{bernoulliTBSS} has slightly lower power than the other methods. 

\figurename{s}~\ref{fig:fdr}~and~\ref{fig:tpr} show the proportion of false and true positive results, respectively. These metrics provide deeper insights into the properties of the considered methods. For example, although for a large value of the exposure effect, \textit{poissonTBSS} has a global power of $1,$ this is mostly due to erroneous identification of nodes in $\mathcal{I}_0$ rather than detection of true positives. For example, in row 1 column 1, the largest proportion of true positive results detected by \textit{poissonTBSS} is approximately $89\%$ while the remaining $11\%$ are false positive results. In the same setting, \textit{exponentialTBSS}, \textit{coxTBSS}, and \textit{robustTBSS} have a proportion of true positive results of approximately $100\%$ and a proportion of false positive results close to $0$.

\figurename~\ref{fig:sim_time} in the supplementary materials shows the computation time for the five methods under consideration. Execution time is primarily driven by the MC algorithm used to generate replicate data under the global null.  The fastest of the considered methods is \textit{poissonTBSS} with an average computational time of $5$ seconds ($s$) (standard deviation (SD) $1.65s$), followed by \textit{bernoulliTBSS} with an average computation time of $7s$  (SD  $2.06s$), \textit{exponentialTBSS} with $11s$ (SD $3.18s$), and \textit{robustTBSS} $12s$ (SD $3.05s$). As expected, the slower method is \textit{coxTBSS}, which requires a permutation algorithm and takes approximately $1324s$ on average (SD $427.50s$). Code to replicate the simulations is available at  \url{https://github.com/rMassimiliano/TBSS_survival_simulation}.


\section{Application}\label{sec:application}
We re-analyze data from a new-user, active-comparator cohort design described in \citet{russo:2025}. The study considers the safety profile of two glucose-lowering agents: the second-generation sulfonylureas (SUs) and dipeptidyl peptidase-4 (DPP-4i). DPP-4i acts as an active comparator, i.e., a drug used in similar settings that serves as the standard of care and helps mitigate bias related to the indication. Using an active comparator can reduce the detection of adverse events related to the underlying medical condition for which a drug is prescribed, rather than to the drug itself. Data for the study were extracted from Electronic Health Records of the Mass General Brigham hospital system between 2000 and 2020, linked to claims from fee-for-service Medicare (2007--2020), curated by the FDA Sentinel Real World Evidence Data Enterprise Initiative~\citep{desai:2024}. Potential adverse events have been extracted using diagnosis (ICD) codes, clinical note analysis via a Natural Language Processing algorithm~\citep{smith:2026}, and laboratory results (binarized using thresholds from clinical practice).
The study design is described in detail in \figurename~\ref{fig:study_design} and its caption. We also refer the reader to ~\citet{russo:2025} for a more detailed description. The design establishes a clear temporal sequence from treatment exposure to outcomes and enables capturing events occurring shortly after treatment initiation. In the analysis proposed here, we consider all adverse events recorded in inpatients and emergency settings. The outcomes are hierarchically related using the MedDRA hierarchy and comprise $2162$ nodes with at least one observed event, of which $695$ have at least two events for the exposure group (SUs users), at least one event for the comparator group (DPP-4i), and a larger rate of events for SUs than DPP-4i users.

There are two main differences with the modeling approach discussed in \secname~\ref{sec:simulations}. First, because of the limited sample size, we balanced initiators of each drug using propensity score fine stratification weighting~\citep {desai:2017} rather than matching. Fine stratification divides study participants into several strata (e.g., 50 or 100 strata) using quantiles of an estimated propensity score. Using the PS distribution among the exposed as reference, for each stratum, exposed individuals (SUs users) are weighted at 1, while ``unexposed'' individuals (DPP-4i users) receive a weight proportional to the ratio of exposed to unexposed individuals in that stratum.  In our context, if we consider $S$ strata and let $s_{i} \in \{1, \ldots, S\}$ be an indicator for the stratum for a generic individual $i$, the weights are defined as 
\begin{equation}
w_i  = 
  \indicator\{A_i =1\} + \indicator\{A_i =0\}
\frac{ N_{1,s_i} /N_{1} }{ N_{0,s_i} / N_{0} }
\label{eq:weights}
\end{equation}
where $N_{a,s_i}$ is the number of individuals in exposure group $a \in \{0,1\}$ and stratum $s_i\in \{1,\ldots, S\}$, while $N_a$ is the total number of individuals in exposure group $a$. If the propensity score accurately captures all the confounding factors,  weighted estimates from this approach consistently estimate the average treatment effect among the exposed~\citep{desai:2017}.

Second, in \cite{russo:2025}, not all the events recorded in the database are  ``incident'' and counted as adverse events for the analysis. Refer to \secname~\ref{sec:incident}  of the supplementary material. Specifically, only the events for study participants who did not experience the same events within the 180 days prior to the index date are considered incident outcomes.  We apply this definition separately to each node of the tree.

We adapt our \textit{coxTBSS} algorithm of Section~\ref{sec:cox} to analyze this data.  For fine stratification, in Step 1 of \algname~\ref{alg:coxTbss} we permute individuals within the same PS-stratum which, under a correctly specified PS-model, are approximately exchangeable, and then compute a weighted partial likelihood ratio test using the fine-stratification weights:
\begin{equation*}
\logL^{\mbox{\tiny  w-cox}}_g (\HR_g) = 
\sum_{i: \delta^{(g)}_i=1} w^{(g)}_i
{A^{(g)}_i}
\log\{
\HR_g\} -
\sum_{i: \delta^{(g)}_i=1} w^{(g)}_i
\log\left\{ \sum_{\{j: U^{(g)}_j \geq U^{(g)}_i\}} w^{(g)}_j\HR^{A^{(g)}_j}_g\right\},\mbox { for }
g\in \mathcal{G},
\end{equation*}
where $\widehat \HR_g = \arg \max_{\varphi > 0} \logL^{\mbox{\tiny  w-cox}}_g (\HR_g)$. We use $\logL^{\mbox{\tiny  w-cox}}_g$ as a statistic for TBSS, given by
$T^{\mbox{\tiny  w-cox}}_g = 
2\times\{
\logL^{\mbox{\tiny  w-cox}}_g (\widehat \HR_g)
- 
\logL^{\mbox{\tiny  w-cox}}_g (1)
\}.$ 

To account for the node-specific definition of incident outcomes, we appropriately define censoring indicators for each node in the tree. Because a claim diagnosis for a certain individual might be considered in the analysis or ignored depending on the node of the tree, the algorithms described in Sections~\ref{sec:exp}~and~\ref{sec:robust} are not immediately adaptable to this specific analysis. As a comparator, we consider a weighted exponential TBSS in which we compute replicates under the null using the same permutation algorithm as in \textit{coxTBSS}. The weighted exponential model uses the likelihood ratio in Equation~\eqref{eq:lrt_exp} with weighted statistics 
$R^{(g)}_a = \sum_{i=1}^{n_g} w^{(g)}_i\delta^{(g)}_i \indicator\{A^{(g)}_i=a\}$ and $W^{(g)}_a = \sum_{i=1}^{n_g} w^{(g)}_i U^{(g)}_i \indicator\{A^{(g)}_i=a\}$. This comparator uses the same statistic as in \citet{russo:2025} with a different weighting scheme. \citet{russo:2025} scaled the weights for the unaffected by multiplying by the ratio of exposed to unexposed.  Here, we use the definition of Equation~\eqref{eq:weights}.

In both analyses, an individual can experience multiple events. Because we permute exposure groups, this does not raise concerns about the permutation algorithm's validity, since each individual carries their full event history. However, to obtain a sensible estimate of node-specific statistics, we consider only one event (the earliest) for each node and individual. This avoids the double-counting of events.

\tablename~\ref{tab:app_res} shows the results for 11 outcome nodes. The reported nodes have been selected as the union of the ten nodes with the largest observed LRTs in the two analyses (that coincide) and the node  `Hypoglycaemia', an anticipated signal from previous studies~\citep{mishriky:2015}. The two methods yield consistent results in this analysis, with all nodes sharing the most extreme test statistics. None of the analyses show any alert (p-value $\leq0.05$). For \textit{exponentialTBSS}, the smallest p-value ($\sim0.07$) corresponds to ``Headache'', which has been identified as a potential alert in previous analyses~\citep{russo:2025}. For the \textit{coxTBSS}, the smallest p-value ($\sim0.06$) is for the same ``Headache'' MedDRA node. The second and third nodes also refer to other definitions of headaches, which are nonspecific symptoms of hypoglycemia, and plausible alerts in this analysis.

The lack of signal---particularly the failure to identify increased risk of hypoglycemia---most likely reflects the small sample size of this safety surveillance study (2820 new initiators of SU and 985 new initiators of DPP-4i). 

\begin{table}[h!]
\centering
\resizebox{\textwidth}{!}{
{\tiny
\begin{tblr}{colspec = {X[4,l,m] *{7}{X[l]}},
  		row{even} = {black!2!white},
  		row{odd} = {black!10!white},
        row{1,2} = {black!25!white,font=\bfseries},
        width=\textwidth}
Node description & MedDRA & Observed & Relative & Exp. & Exp. & Cox & Cox \\
& Level &  rate $\times$ 1000 & Risk & LRT & p-value & LRT & p-value \\
 
Headaches &  HLGT & 0.182 & 6.300 &  10.343 & 0.069 & 20.717 & 0.063 \\ 
Headache &  PT & 0.182 & 4.809 &  8.694 & 0.165 & 17.397 & 0.152 \\ 
Headaches NEC &  HLT & 0.179 & 4.736 &  8.465 & 0.183 & 16.957 & 0.176 \\ 
Death and sudden death &  HLT & 0.060 & 28.514 &  5.478 & 0.705 & 10.981 & 0.695 \\ 
Fatal outcomes &  HLGT & 0.060 & 28.514 &  5.478 & 0.705 & 10.981 & 0.695 \\ 
Death &  PT & 0.057 & 27.203 &  5.189 & 0.768 & 10.395 & 0.757 \\ 
Gastrointestinal signs and symptoms &  HLGT & 0.178 & 2.508 &  4.104 & 0.958 & 8.227 & 0.964 \\ 
Gait disturbance &  PT & 0.090 & 4.405 &  4.070 & 0.961 & 8.188 & 0.965 \\ 
Gait disturbances &  HLT & 0.098 & 3.731 &  3.835 & 0.978 & 7.723 & 0.981 \\ 
Gastrointestinal and abdominal pains (excl oral and throat) &  HLT & 0.197 & 2.216 &  3.711 & 0.986 & 7.498 & 0.985 \\ 
Hypoglycaemia &  PT & 0.095 & 2.174 &  1.738 & 1.000 & 3.448 & 1.000 
\end{tblr}}
}
\caption{Likelihood ratio tests and p-values for the permutation-based \textit{ExponentialTBSS} and \textit{CoxTBSS} described in \secname~\ref{sec:application}. The observed rates represent the weighted number of SUs divided by the weighted follow-up time multiplied by $1000$. We used the weights defined in \eqref{eq:weights}. MedDRA level abbreviations are:
High Level Group Term (HLGT), High Level Term (HLT), Preferred Term (PT).
} 
\label{tab:app_res}
\end{table}


\section{Discussion}\label{sec:discussion}
For disproportionality analysis based on database studies, time-to-event data are typically available. These data have been used to improve comparability of follow-up duration between the exposure group, but are often ignored at the analysis stage. We introduced and compared three methods that leverage time-to-event data for TBSS signal detection. Across various simulation scenarios, \textit{robustTBSS} offers a good trade-off between computational efficiency and signal-detection power. Additionally, it requires only a few data summaries, facilitating analysis when patient-level data are not readily accessible due to privacy or other constraints. 

Importantly, as previously noted in \citet{thuy:2024}, our simulations suggest that ignoring uncertainty in the estimates of the expected number of events in the traditional \textit{poissonTBSS} method---the current practice in Pharmacoepidemiology---might lead to a severe inflation of type-I error even in settings where no other biases influence the analysis. The magnitude of this inflation depends on the setting. The \textit{exponentialTBSS} algorithm proposed in this paper can be considered as a generalization of \textit{poissonTBSS} that appropriately accounts for uncertainty in the number of events for the comparator group. Since for comparative safety studies, \textit{poissonTBSS} and \textit{exponentialTBSS} require the same summary statistics, \textit{exponentialTBSS} should be preferred, particularly in settings with small sample sizes.

Our methods have some limitations. For example, in describing the methods, we followed the TBSS structure used in practice~\citep{kulldorff:2003}, which assumes that data at leaf-level nodes can be treated as realizations of independent random variables. This can be a limitation in studies where patients can indeed experience multiple events, and therefore be associated with multiple leaves. As noted in \secname~\ref{sec:application}, \textit{coxTBSS} can be directly adapted to this case by permuting individual exposure-group allocations and using appropriate test statistics. Co-occurring outcomes are still grouped together in each permutation,  reflecting potential biological or clinical relationships. In our case, for each node test statistic we used the first-occurring incident event for each outcome node when fitting the weighted Cox model, but other approaches that directly model the occurrence of multiple events can be considered.  In contrast, the methods proposed in \secname{s}~\ref{sec:exp} and \ref{sec:robust}, like other TBSSs, directly use the assumption that the leaf levels of nodes are independent. In some cases~\citep{huybrechts:2021}, it is possible to create an artificial leaf level in the `tree' that comprises only independent events.  These new leaf nodes should include summaries of all unique combinations of events observed in a given study. Replicates under the global null can then be generated from this artificial leaf, with the caveat that when summary statistics are aggregated, individuals should not be double-counted.  However, how to integrate node-specific washout into the methods proposed in \secname{s}~\ref{sec:exp} and \ref{sec:robust}, as well as the `traditional' \textit{poissonTBSS}, to avoid a permutation algorithm, remains an open question.

\section*{Acknowledgments}
This work has been supported by NIH grants R01HD110092, R01HD104646, R03AI178363. Computation was carried out on the Unity Cluster of the College of Arts and Sciences at the Ohio State University. The computational resources provided is gratefully acknowledged. Grammarly was used for grammar checking.

\newcommand{\beginsupplement}{%
        \setcounter{table}{0}
        \renewcommand{\thetable}{S\arabic{table}}%
        \setcounter{figure}{0}
        \renewcommand{\thefigure}{S\arabic{figure}}%
	\setcounter{section}{0}
        \renewcommand{\thesection}{S\arabic{section}}
     }

\section*{Supplementary Material}
\beginsupplement

\section{Proof of Proposition 1}

\begin{assumption}\label{as:coherence_S}
For all nodes $g \in \mathcal{G},$  $H_{0g}$ holds if and only if  $\cap_{\ell \in \mathcal L_g} H_{0\ell}$ holds. In other words, a null hypothesis holds if and only if all the descendants' null hypotheses also hold.
\end{assumption}

\begin{prop}\label{prop:strong_FWER_S}
When  Assumption \ref{as:coherence_S} holds,
 the decisions $\{\indicator\{{\tt pv}_{g} \leq \alpha\}\}_{g \in \mathcal G}$ to reject node-specific null hypotheses $\{H_{0g}\}_{g \in \mathcal{G}}$ controls the FWER for all sets $\mathcal I_0 \subseteq \mathcal G$.
\end{prop}
\begin{proof}
Let $\mathcal I_0 \subseteq \mathcal{G}$ be the unknown set of true null hypotheses and $\mathcal L(\mathcal I_0) = \{ \cap_{g \in \mathcal I_0} \mathcal L_g  \}$ the indices of the corresponding leaves. Without loss of generality, we assume that the random variable $T= \max_{g \in \mathcal{G} } T_g$ is continuous and let $
c_\alpha$ be the $\alpha$-level quantile of the distribution of $T $ under the global null, that is $c_\alpha$ satisfies
$
\pr\{T \geq c_\alpha | H_0\} = \alpha. 
$
By definition the decision $\indicator\{{\tt pv}_{g} \leq \alpha\}$ is equivalent to $\indicator\{T_{g} \geq c_\alpha\}.$

Under Assumption \ref{as:coherence_S}, we have that: (1) $\{\cap_{g \in \mathcal I_0} H_{0g}\}$ is equivalent to $\{\cap_{\ell \in \mathcal L(\mathcal I_0)} H_{0\ell}\},$ (2)
because the data in different leaves are independent, the joint distribution of $\{T_g\}_{g \in \mathcal I_0}$ is independent of the distribution of $\{T_\ell\}_{\ell \notin \mathcal L(\mathcal I_0)}.$ Using these facts, we can write the FWER as
\begin{align*}
\mbox{FWER}(\mathcal I_0) &= \pr\{\mbox{Reject at least one } H_{0g} \mbox{ for } g \in \mathcal I_0 | \cap_{g \in \mathcal I_0} H_{0g} \} \\
&= \pr\{ \cup_{g \in \mathcal I_0}\indicator\{{\tt pv}_g \leq \alpha\}   | \cap_{g \in \mathcal I_0} H_{0g} \} \\
&=  \pr\{ \cup_{g \in \mathcal I_0}\indicator\{T_g  \geq c_\alpha \}  | \cap_{g \in \mathcal I_0} H_{0g} \}  \\
&= \pr\{ \max_{g \in \mathcal I_0}{T_g}  \geq c_\alpha  | \cap_{g \in \mathcal I_0} H_{0g} \}  \\
&=\pr\{ \max_{g \in \mathcal I_0}{T_g}  \geq c_\alpha  | \cap_{\ell \in \mathcal L(\mathcal I_0)} H_{0\ell} \}  &\mbox{ (using (1))}\\
&=\pr\{ \max_{g \in \mathcal I_0}{T_g}  \geq c_\alpha  |  \cap_{\ell \in \mathcal L} H_{0\ell} \}  &\mbox{ (using (2))}\\
&\leq  \pr\{ \max_{g \in \mathcal{G}}{T_g}  \geq c_\alpha  |  \cap_{\ell \in \mathcal L} H_{0\ell} \}  &(\mbox{because } I_0 \subseteq \mathcal G) \\
&\leq  \alpha . & \mbox{ (defintion of } c_\alpha) 
\end{align*}
Control of the FWER for any $\mathcal I_0 \subseteq \mathcal G$ directly follows from the fact that the set $\mathcal I_0$ is arbitrary.
\end{proof}
\section{Supplemental tables and figures}
\begin{table}[ht!]
\centering
\resizebox{0.8\textwidth}{!}{
    \begin{tblr}{
  colspec = {ccl},
  row{even} = {black!2!white},
  row{odd} = {black!10!white},
  row{1} = {black!25!white}
}  
{\bf Aspect} & {\bf Parameters} & { \bf Details} \\
{Number of individual \\for each leaf}& $n^{(\ell)}$&  $n^{(\ell)}=600,$ for a total sample size of $7800$ \\
{Measured  \\ and \\ unmeasured \\confounders} & $F_{XZ}$ &{Two measured confounders:\\
 $\circ$ $X_{1} \in \{0,1\}$ with $\mbox{pr}\{X_1 =1\} = 0.8$  \\
   $\circ$ $X_{2}$ uniform in $ [0,1]$. \\
and one unmeasured confounder\\
$\circ$ $Z_{1} \in \{0,1\}$ with $\mbox{pr}\{Z_1 =1\} = 0.7$  
}
\\
{Treatment assignment \\ PS model}& 
{
$\eta_1$ \\
$\boldsymbol \zeta_1$   \\
$\gamma_1$
}
& 
{The intercept $\eta_1$ varies across scenario. \\
It is selected to have, on average, \\ $50\%$ exposed and $50\%$ unexposed. \\
coefficients for $X_1$ and $X_2$
$\boldsymbol \zeta_1 = (1,1)$\\
The coefficient for unmeasured confounding $Z$\\ $\gamma_1 =0$ for the scenarios correctly specified PS-model, and \\$\gamma_1 =0.7$ for mispecified PS-model.
}
\\
{Time to event \\ generation} & 
{
$\kappa^{(\ell)}_a$\\
$\lambda_0$\\
$\boldsymbol{\zeta}_0$\\
$\boldsymbol  \gamma_0$ }&
{
$\circ$ $\kappa^{(\ell)}_a=1$ for  exponentially distributed and time-delayed time to event scenarios \\
\phantom{$\circ$} and $\kappa^{(\ell)}_a=2$ for Weibull time to events\\
$\circ$ baseline scale $\lambda_0=2$\\
$\circ$   $\boldsymbol \zeta_0 =(0.5,0.5)$ \\
$\circ$   $\gamma_0= 0.3$ \\
}
\\
Censoring & $p_a$ &  $p_a = 0.2,$ for $ a\in \{0,1\}$\\
Treatment effect & $\omega^{(\ell)}$&
{
For the first two treatment scenarios \\ 
we induce an effect in leave $25$ \\
$\circ$  $\omega^{(25)} \in \{0,\cdots,\log(2.5)\},$ and $0$ otherwise.\\
For the third simulation scenario, we let \\
$\circ$  $ \omega^{(13)} = \omega^{(24)} =  \omega^{(25)}   \in \{0,\cdots,\log(2.5)\},$ and $0$ otherwise.\\
In this scenario, the parameter $\omega^{(\ell)}$\\ controls the treatment assignment rather than the time-to-event.\\
For the scenario with delayed effect (e.g., third column in Figure 2 of the main manuscript) \\
the parameter $\omega^{(\ell)}$ is greater than $0$ only for patients with $Y_i > \texttt{delay}.$
}\\
Delayed effect &$\texttt{delay}$& 
{$\circ$ $1$ for scenarios with single signal\\ 
$\circ$  $3$ for scenarios with composite signal  } \\
{TBSS MC algorithm \\ number of replicates}& & {9999 for all TBSSs considered in the paper}&\\
{PS caliper \\ }& & We use $0.2$ times the standard deviation of the PS&\\
\end{tblr}
}
\caption{Details of simulation study in \secname~\ref{sec:simulations} of the main manuscript}
    \label{tab:mode_and_deriv}
\end{table}

\begin{figure}[ht!]
    \centering
    \includegraphics[width=\textwidth]{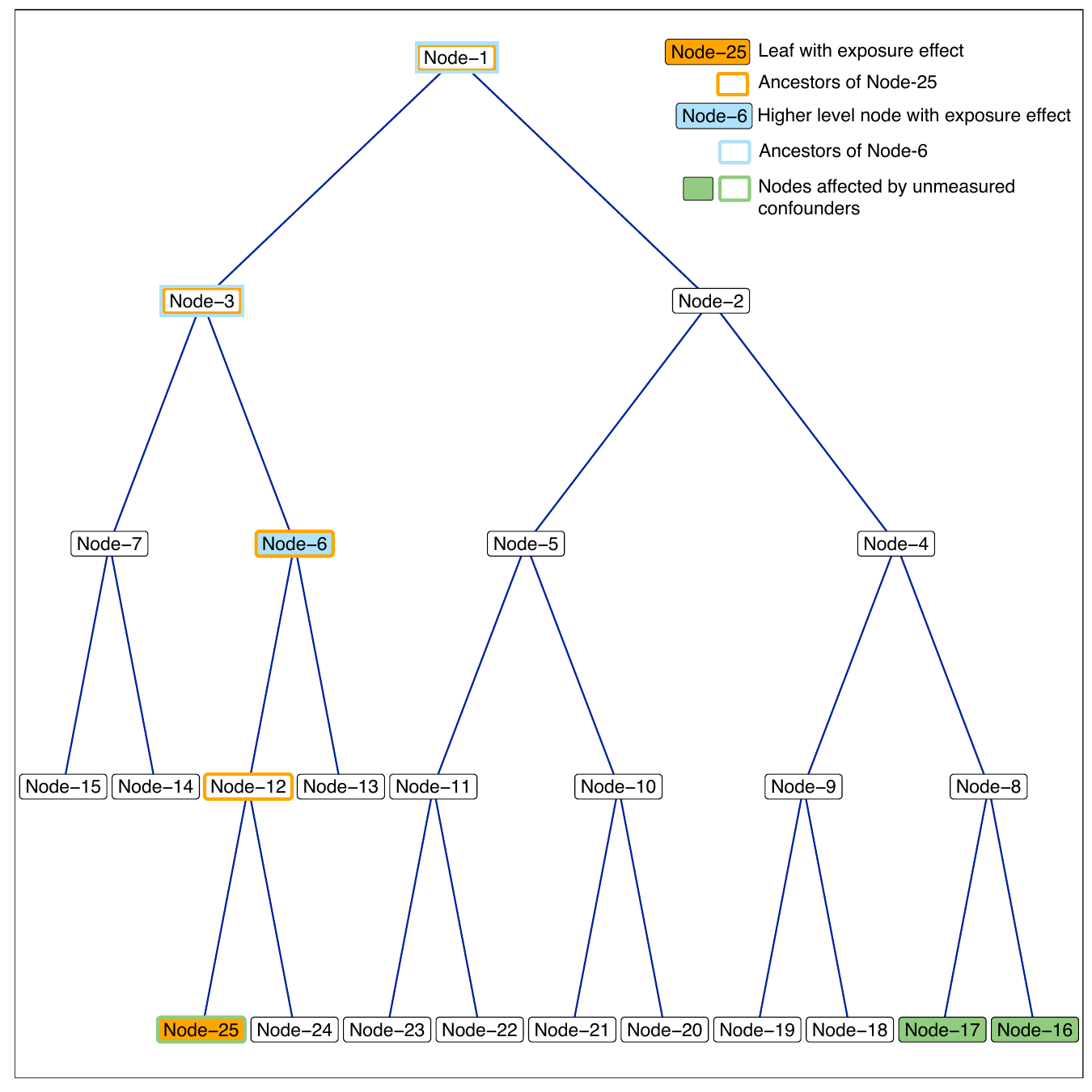}
\caption{Tree used for the simulations presented in \secname~\ref{sec:simulations} of the main manuscript. Node-25 (in orange) is used to inject exposure effects in scenarios where only a leaf is affected, while Node-6 (in blue) is used in scenarios where the exposure effect is injected in a higher-level node. Nodes in the path of Node-25 and Node-6 are circled with the corresponding color. We can expect signal propagation through these nodes. Nodes-16, Nodes-17, and Nodes-25 (in green) are also affected by unmeasured confounding when the PS-model used to create a matched cohort is mispecified.}
    \label{fig:sim_tree}
\end{figure}

\begin{figure}[ht!]
    \centering
    \includegraphics[width=\textwidth]{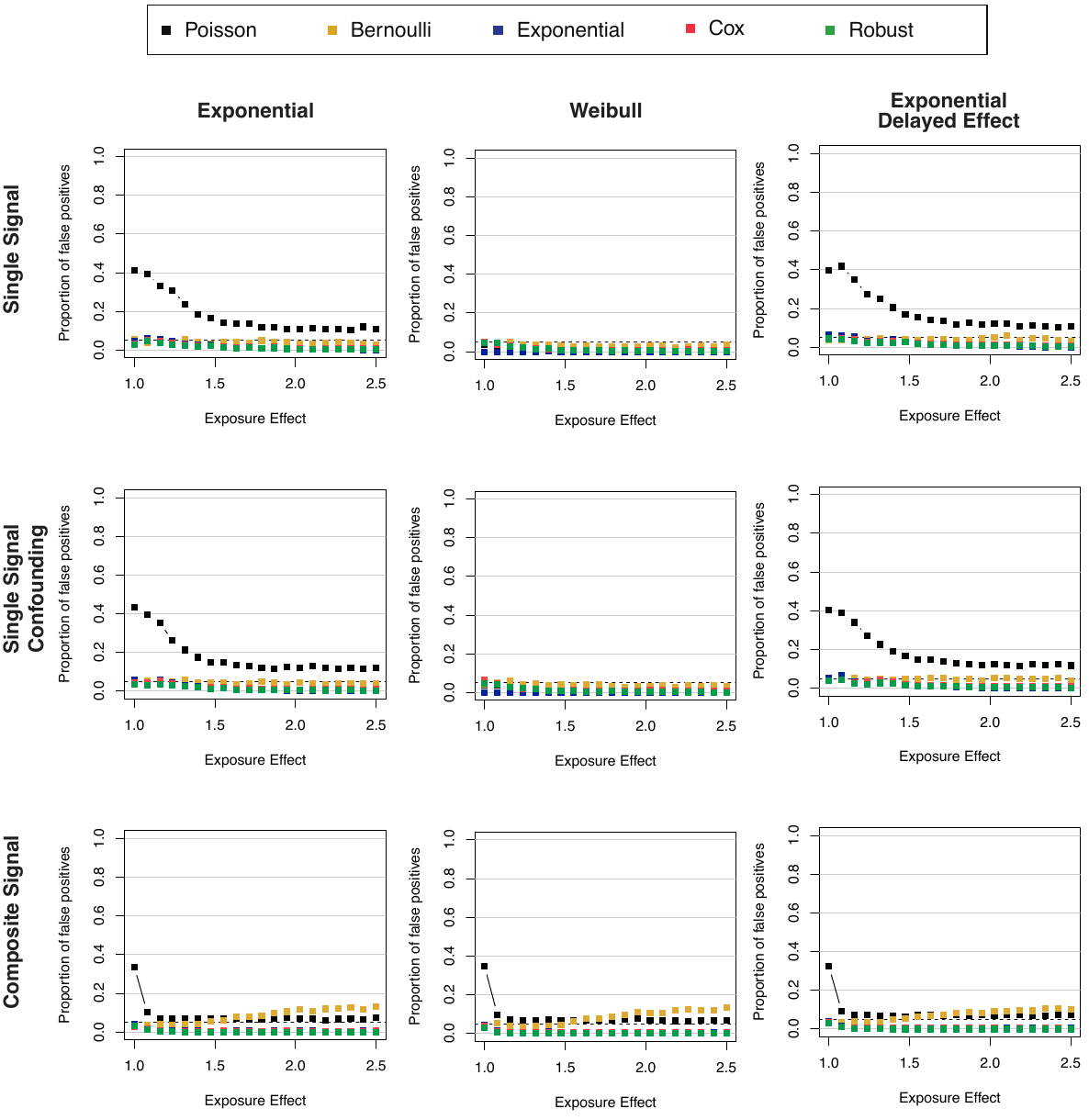}
    \caption{Proportion of nodes in $\mathcal I_0$ rejected (false positive)
	    for the simulations described in \secname~\ref{sec:simulations}. Each column represents a distributional assumption for the time-to-event data: column 1 shows exponential time-to-event data, column 2 shows Weibull time-to-event data, and column 3 shows exponential time-to-event with delayed effect.  The three rows are the three different simulation scenarios. In the first row, the treatment effect is added to a single leaf (Node-25 of the tree in \figurename~\ref{fig:sim_tree}), and the propensity score is correctly specified. In the second row, the propensity score model is mispecified because of unmeasured confounders, and in the third row, the exposure effect is included at a higher level (Node-6  in \figurename~\ref{fig:sim_tree}). The x-axis represents exposure effects $\exp\{\omega^{(25)}\}$ or the effect induced on Node-6 (third row). An exposure effect = 1 indicates no effect.
    }
    \label{fig:fdr}
\end{figure}

\begin{figure}[ht!]
    \centering
    \includegraphics[width=\textwidth]{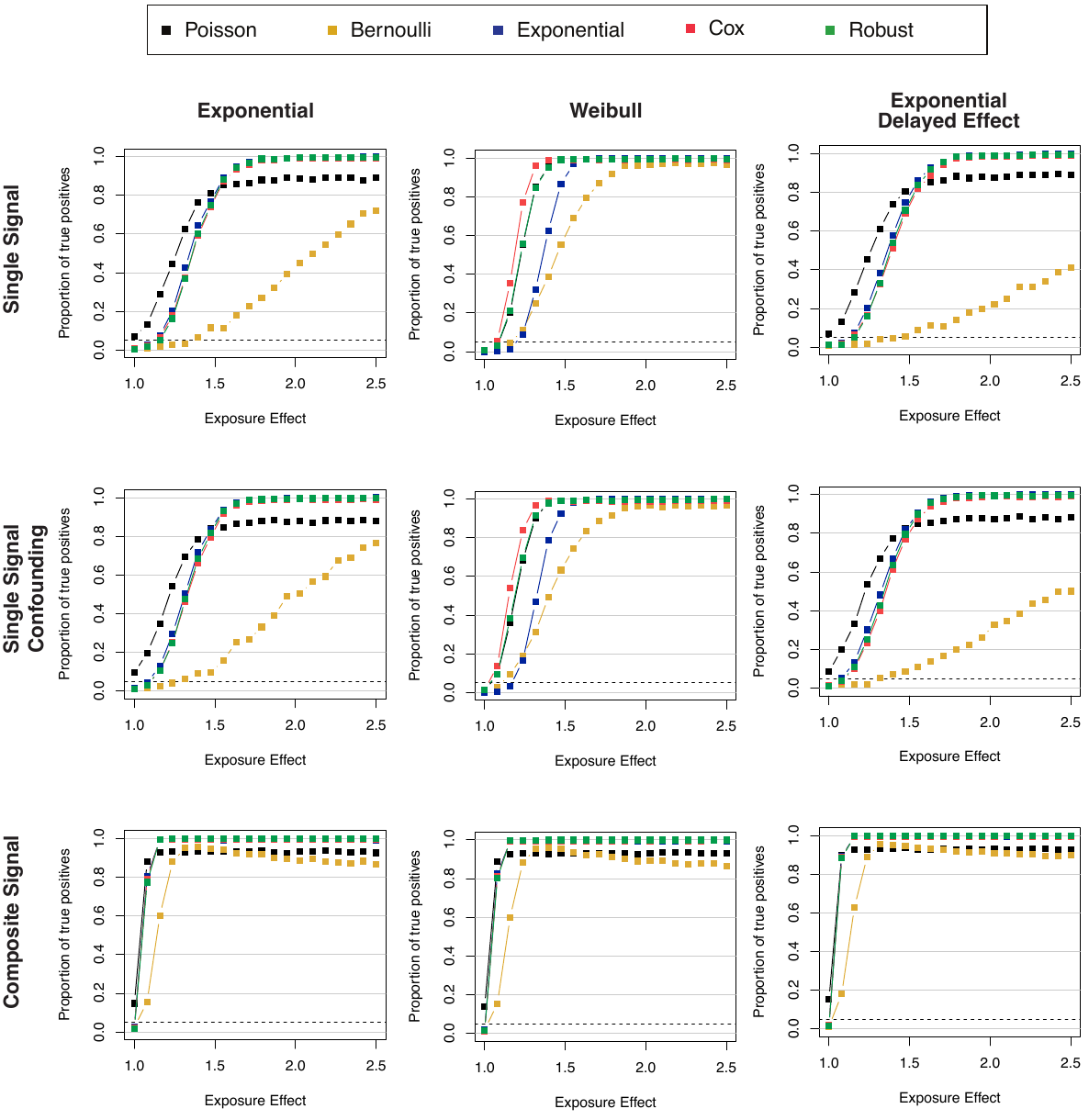}
    \caption{Proportion of nodes in $\mathcal I_1$ rejected (true positive) for the simulations described in \secname~\ref{sec:simulations}. Each column represents a distributional assumption for the time-to-event data: column 1 shows exponential time-to-event data, column 2 shows Weibull time-to-event data, and column 3 shows exponential time-to-event with delayed effect.  The three rows are the three different simulation scenarios. In the first row, the treatment effect is added to a single leaf (Node-25 of the tree in \figurename~\ref{fig:sim_tree}), and the propensity score is correctly specified. In the second row, the propensity score model is mispecified because of unmeasured confounders, and in the third row, the exposure effect is included at a higher level (Node-6  in \figurename~\ref{fig:sim_tree}). The x-axis represents exposure effects $\exp\{\omega^{(25)}\}$ or the effect induced on Node-6 (third row). An exposure effect = 1 indicates no effect.}
    \label{fig:tpr}
\end{figure}

\begin{figure}[ht!]
    \centering
    \includegraphics[width=\textwidth]{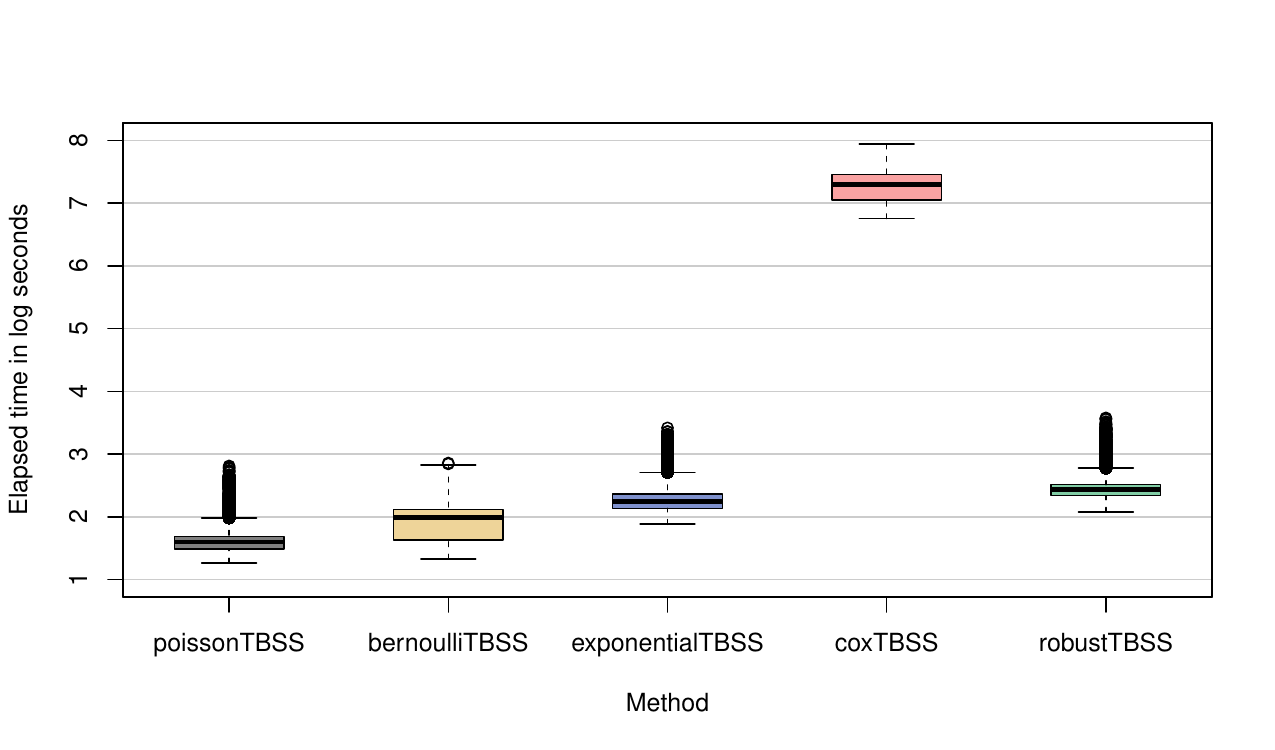}
\caption{Boxplots for the computational times of the methods considered in \secname~\ref{sec:simulations} of the main manuscript. Times are elapsed times computed using the \texttt{system.time()} function in R (third output) expressed on the log-second scale. The boxplots show the execution times across all simulation scenarios, with 160'000 replicates for each method.}
    \label{fig:sim_time}
\end{figure}

\begin{figure}[ht!]
\includegraphics[width=\textwidth]{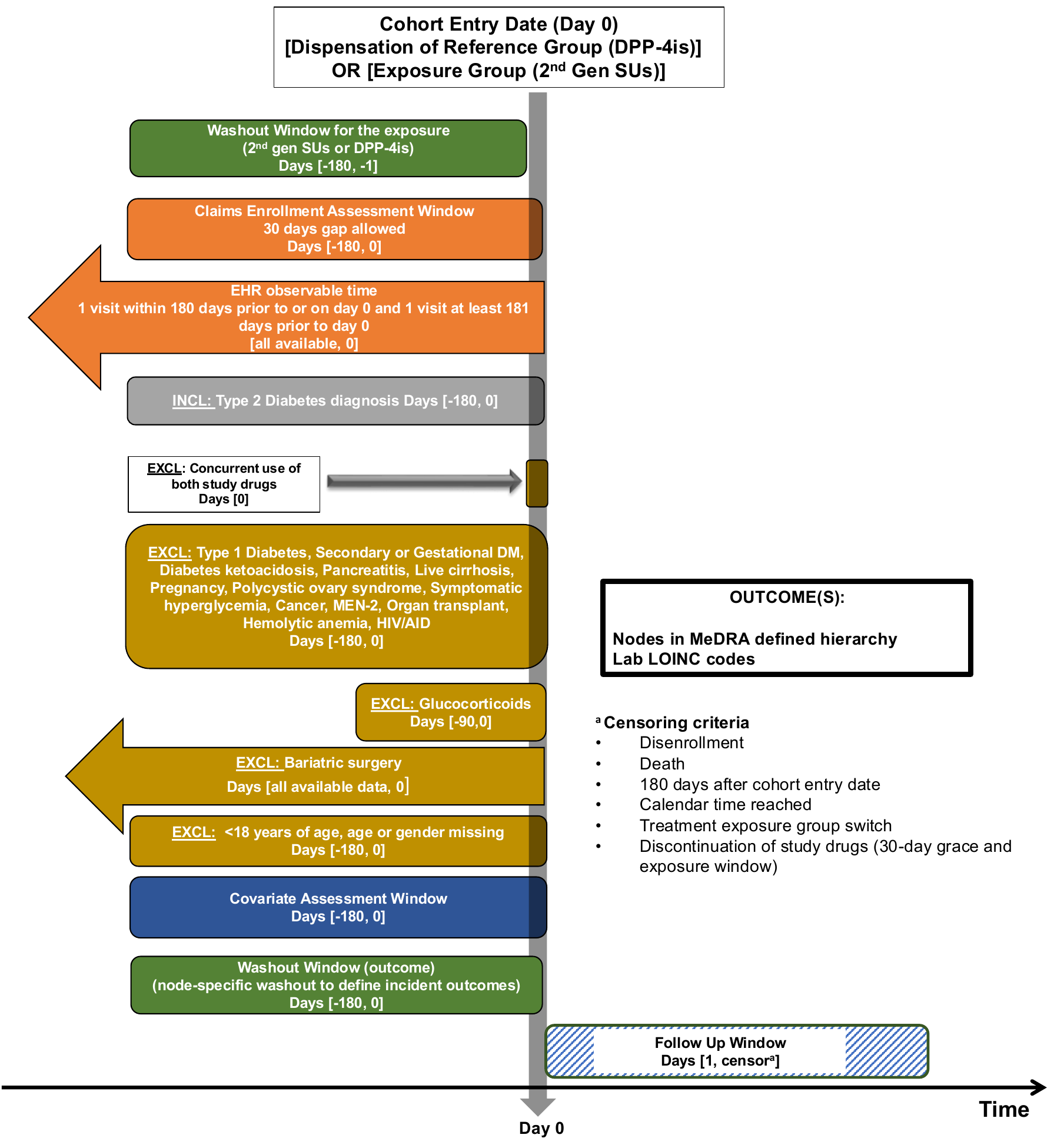}
\caption{Diagram of the study design for the application considered in \secname~\ref{sec:application} of the main manuscript.}
\label{fig:study_design}
\end{figure}

\clearpage
\section{Additional details on the simulation study}
\subsection{Propensity score matching}\label{sec:ps_matching}
We consider the following procedure~\citep[][]{rosenbaum:1983}:
\begin{enumerate}
    \item Fit a propensity score model via logistic regression with covariates $\boldsymbol X^{(\ell)}_i,$ for $i=1,\ldots,n_\ell$ and $\ell \in\mathcal{L},$
    $\texttt{PS}^{(\ell)}_i =  [1 + \exp\{-(\eta_1 +[\textbf {X}^{(\ell)}_i]^\intercal  \boldsymbol \zeta_1   )\}]^{-1}.$
    We use a single PS model for all leaves of the tree (see \cite{wang:2018} for discussion). In our simulations without unmeasured confounders, $\boldsymbol{\gamma}_1 = (0, \ldots,0)$, this is a correctly specified PS model.

    \item For each treated individual in the dataset, we compute the absolute difference of their PS with each non-treated individual. We collect this distance in a $\sum_{\ell \in \mathcal{L}} n^{(\ell)}_1 \times \sum_{\ell \in \mathcal{L}} n^{(\ell)}_0$ dimensional matrix $\mathbf M$.

    \item For each individual, we find the best match: an individual in the other exposure group with the smallest absolute propensity score distance. We consider the matches in order of distance, i.e., the first matched pair is the pair of individuals with the smallest observed PS absolute distance (``best match''). The second matched pair is the pair of individuals with the second smallest observed PS absolute distance (``second best match''), and so on. 

\end{enumerate}
The matching is without replacement, and a match is considered valid if and only if the distance (collected in the matrix $\mathbf M $) is smaller than a prespecified {\it caliper}.  Individuals who are not matched are discarded from the analysis.

\subsection{Generate simulated data with composite signals}
\label{sec:composite}

For simulation scenarios with composite signals (e.g., the third row in \figurename~\ref{fig:sim1}  of the main manuscript), we inject a signal by adjusting the probability that a patient is assigned to a certain exposure group based on their time-to-event in the composite node of interest. In other words, by associating patients with larger/smaller time-to-event distributions within the exposure group, we implicitly suggest better/worse outcomes for these patients. 

As in the other simulations, we start by generating time-to-events
$Y^{(\ell)}_i$  
and covariates 
${X}^{(\ell)}_i$  and  $\mathbf Z^{(\ell)}_i$
for  $i=1,\ldots, n_\ell$ and $\ell \in \mathcal L$ from pre-specified distributions.  We define the propensity score associated with patient $i$ as:
$$
\texttt{PS}^{(\ell)}_i =  [1 + \exp\{-(\eta_1 +[\textbf {X}^{(\ell)}_i]^\intercal  \boldsymbol \zeta_1   +  [\mathbf Z^{(\ell)}_i]^{\intercal} \boldsymbol \gamma_1)\}]^{-1}.
$$

Treatment effect is then induced by using a probability of being assigned to the exposure group defined via
$$
\mbox{pr}\{A^{(\ell)}_i = 1\} = [1 + \exp\{-(\upsilon_0  + \texttt{PS}_i  + \omega^{(\ell)}Y^{(\ell)}_i)\}]^{-1}.
$$

The parameter $\upsilon_0$ is chosen to ensure that approximately 50\% of the patients are associated with each exposure group. For positive $\omega^{(\ell)}$,  patients with larger time-to-event are more likely to be associated with the $A^{(\ell)}_i = 1$ group.  In our simulation we set $\omega^{(\ell)} \in \{0,\cdots,\log(2.5)\},$  for the leaves connected with Node-6 (Nodes 13, 24, and 25, seed \figurename~\ref{fig:sim_tree}), and $0$ for all the other leaves. 

For the scenario with delayed exposure effect (row 3, column 3) the parameter  $\omega^{(\ell)}$  for $\ell \in \{13,24,25\}$ is different from $0$ only for patients where $Y_i> 3.$

\section{Incident outcomes}\label{sec:incident}
TBSSs are typically used to screen and prioritize statistical alerts. When screening potential outcomes associated with medications of interest, it is important to focus on outcomes that the medications could theoretically cause. Therefore, diagnostic codes reflecting things like ``family history'' or conditions with long latency/induction, like cancer, may be pruned from the tree. Additionally, to be more confident that diagnostic codes reflect new occurrences of a health event, TBSS analyses typically focus on incident outcomes.

Operationally, a washout window (e.g., six months) prior to the index date and a tree level such  
Level 3 of the tree (e.g., the first three characters of an ICD-10 code) are chosen. Then, an outcome at the selected tree level observed during the follow-up window is considered incident if the corresponding patient did not experience the same outcome within the washout window. Diagnosis codes during follow-up that do not meet incident outcome criteria are ignored.  Nodes above the selected level where incident outcomes are defined (e.g., level 1 and 2 in ICD-10) are typically not tested to avoid double counting. Patients may contribute multiple unrelated events in the tree, as long as the corresponding tree-nodes do not eventually meet in the hierarchy. 

An alternative washout strategy based on node-by-node definition of incident outcomes has been proposed in~\citet{russo:2025}. This is helpful for complex multiaxial trees, such as MedDRA, where defining incidence at a level is problematic because there are multiple overlapping pathways up the tree hierarchy. However, because each patient might contribute differently at each node, computing p-values requires permuting patient-level exposure.
We apply this strategy in 
\secname~\ref{sec:application} of the main manuscript.

\bibliography{refs}

@article{kulldorff:2003,
	title = {A {Tree}-{Based} {Scan} {Statistic} for {Database} {Disease} {Surveillance}},
	volume = {59},
	issn = {0006341X},
	url = {https://onlinelibrary.wiley.com/doi/10.1111/1541-0420.00039},
	doi = {10.1111/1541-0420.00039},
	language = {en},
	number = {2},
	urldate = {2023-01-17},
	journal = {Biometrics},
	author = {Kulldorff, Martin and Fang, Zixing and Walsh, Stephen J},
	month = jun,
	year = {2003},
	note = {Number: 2},
	pages = {323--331},
}

@incollection{glaz:2001,
	address = {New York, NY},
	title = {Scan {Statistics}},
	isbn = {978-1-4419-3167-2 978-1-4757-3460-7},
	url = {http://link.springer.com/10.1007/978-1-4757-3460-7_1},
	urldate = {2023-06-26},
	booktitle = {Scan {Statistics}},
	publisher = {Springer New York},
	author = {Glaz, Joseph and Naus, Joseph and Wallenstein, Sylvan},
	collaborator = {Glaz, Joseph and Naus, Joseph and Wallenstein, Sylvan},
	year = {2001},
	doi = {10.1007/978-1-4757-3460-7_1},
	note = {Series Title: Springer Series in Statistics},
	pages = {3--9},
}

@article{fralick:2021,
	title = {A novel data mining application to detect safety signals for newly approved medications in routine care of patients with diabetes},
	volume = {4},
	issn = {2398-9238, 2398-9238},
	url = {https://onlinelibrary.wiley.com/doi/10.1002/edm2.237},
	doi = {10.1002/edm2.237},
	language = {en},
	number = {3},
	urldate = {2023-07-03},
	journal = {Endocrinology, Diabetes \& Metabolism},
	author = {Fralick, Michael and Kulldorff, Martin and Redelmeier, Donald and Wang, Shirley V. and Vine, Seanna and Schneeweiss, Sebastian and Patorno, Elisabetta},
	month = jul,
	year = {2021},
}

@article{wang:2018,
	title = {Data {Mining} for {Adverse} {Drug} {Events} {With} a {Propensity} {Score}-matched {Tree}-based {Scan} {Statistic}},
	volume = {29},
	issn = {1044-3983},
	url = {https://journals.lww.com/00001648-201811000-00019},
	doi = {10.1097/EDE.0000000000000907},
	language = {en},
	number = {6},
	urldate = {2023-07-04},
	journal = {Epidemiology},
	author = {Wang, Shirley V. and Maro, Judith C. and Baro, Elande and Izem, Rima and Dashevsky, Inna and Rogers, James R. and Nguyen, Michael and Gagne, Joshua J. and Patorno, Elisabetta and Huybrechts, Krista F. and Major, Jacqueline M. and Zhou, Esther and Reidy, Megan and Cosgrove, Austin and Schneeweiss, Sebastian and Kulldorff, Martin},
	month = nov,
	year = {2018},
	pages = {895--903},
}

@article{huybrechts:2021,
	title = {Active {Surveillance} of the {Safety} of {Medications} {Used} {During} {Pregnancy}},
	volume = {190},
	issn = {0002-9262, 1476-6256},
	url = {https://academic.oup.com/aje/article/190/6/1159/6071894},
	doi = {10.1093/aje/kwaa288},
	language = {en},
	number = {6},
	urldate = {2023-08-16},
	journal = {American Journal of Epidemiology},
	author = {Huybrechts, Krista F and Kulldorff, Martin and Hernández-Díaz, Sonia and Bateman, Brian T and Zhu, Yanmin and Mogun, Helen and Wang, Shirley V},
	month = jun,
	year = {2021},
	pages = {1159--1168},
}

@article{kulldorff:2013,
author = {Kulldorff, Martin and Dashevsky, Inna and Avery, Taliser R. and Chan, Arnold K. and Davis, Robert L. and Graham, David and Platt, Richard and Andrade, Susan E and Boudreau, Denise and Gunter, Margaret J. and Herrinton, Lisa J. and Pawloski, Pamala A. and Raebel, Marsha A. and Roblin, Douglas and Brown, Jeffrey S.},
title = {Drug safety data mining with a tree-based scan statistic},
journal = {Pharmacoepidemiology and Drug Safety},
volume = {22},
number = {5},
pages = {517-523},
doi = {https://doi.org/10.1002/pds.3423},
url = {https://onlinelibrary.wiley.com/doi/abs/10.1002/pds.3423},
eprint = {https://onlinelibrary.wiley.com/doi/pdf/10.1002/pds.3423},
year = {2013}
}

@article{brown:2013,
  title={Drug adverse event detection in health plan data using the Gamma Poisson Shrinker and comparison to the tree-based scan statistic},
  author={Brown, Jeffrey S and Petronis, Kenneth R and Bate, Andrew and Zhang, Fang and Dashevsky, Inna and Kulldorff, Martin and Avery, Taliser R and Davis, Robert L and Chan, K Arnold and Andrade, Susan E and others},
  journal={Pharmaceutics},
  volume={5},
  number={1},
  pages={179--200},
  year={2013},
  publisher={MDPI}
}

@article{rosenbaum:1983,
  title={The central role of the propensity score in observational studies for causal effects},
  author={Rosenbaum, Paul R and Rubin, Donald B},
  journal={Biometrika},
  volume={70},
  number={1},
  pages={41--55},
  year={1983},
  publisher={Oxford University Press}
}

@book{kalbfleisch:2002,
  title={The statistical analysis of failure time data},
  author={Kalbfleisch, John D and Prentice, Ross L},
  year={2002},
  publisher={John Wiley \& Sons}
}

@article{hjort:1992,
 ISSN = {03067734, 17515823},
 URL = {http://www.jstor.org/stable/1403683},
 author = {Nils Lid Hjort},
 journal = {International Statistical Review / Revue Internationale de Statistique},
 number = {3},
 pages = {355--387},
 publisher = {[Wiley, International Statistical Institute (ISI)]},
 title = {On Inference in Parametric Survival Data Models},
 urldate = {2024-04-12},
 volume = {60},
 year = {1992}
}

@article{russo:2024,
author = {Russo, Massimiliano and Wang, Shirley V.},
title = {An open-source implementation of tree-based scan statistics},
journal = {Pharmacoepidemiology and Drug Safety},
volume = {33},
number = {3},
pages = {e5765},
doi = {https://doi.org/10.1002/pds.5765},
url = {https://onlinelibrary.wiley.com/doi/abs/10.1002/pds.5765},
eprint = {https://onlinelibrary.wiley.com/doi/pdf/10.1002/pds.5765},
year = {2024}
}

@article{Cox:1972,
 ISSN = {00359246},
 URL = {http://www.jstor.org/stable/2985181},
 author = {D. R. Cox},
 journal = {Journal of the Royal Statistical Society. Series B (Methodological)},
 number = {2},
 pages = {187--220},
 publisher = {[Royal Statistical Society, Wiley]},
 title = {Regression Models and Life-Tables},
 urldate = {2024-04-26},
 volume = {34},
 year = {1972}
}

@article{cox:1975,
 ISSN = {00063444},
 URL = {http://www.jstor.org/stable/2335362},
 author = {D. R. Cox},
 journal = {Biometrika},
 number = {2},
 pages = {269--276},
 publisher = {[Oxford University Press, Biometrika Trust]},
 title = {Partial Likelihood},
 urldate = {2024-05-24},
 volume = {62},
 year = {1975}
}

@article{suarez:2023,
  title={Monitoring drug safety in pregnancy with scan statistics: a comparison of two study designs},
  author={Suarez, Elizabeth A and Nguyen, Michael and Zhang, Di and Zhao, Yueqin and Stojanovic, Danijela and Munoz, Monica and Liedtka, Jane and Anderson, Abby and Liu, Wei and Dashevsky, Inna and others},
  journal={Epidemiology},
  volume={34},
  number={1},
  pages={90--98},
  year={2023},
  publisher={LWW}
}

@book{dickhaus:2014,
	address = {Berlin, Heidelberg},
	title = {Simultaneous {Statistical} {Inference}: {With} {Applications} in the {Life} {Sciences}},
	copyright = {https://www.springernature.com/gp/researchers/text-and-data-mining},
	isbn = {978-3-642-45181-2 978-3-642-45182-9},
	shorttitle = {Simultaneous {Statistical} {Inference}},
	url = {https://link.springer.com/10.1007/978-3-642-45182-9},
	language = {en},
	urldate = {2024-06-11},
	publisher = {Springer Berlin Heidelberg},
	author = {Dickhaus, Thorsten},
	year = {2014},
	doi = {10.1007/978-3-642-45182-9}
}

@article{schachterle:2019,
	title = {An {Implementation} and {Visualization} of the {Tree}-{Based} {Scan} {Statistic} for {Safety} {Event} {Monitoring} in {Longitudinal} {Electronic} {Health} {Data}},
	volume = {42},
	issn = {0114-5916, 1179-1942},
	url = {http://link.springer.com/10.1007/s40264-018-00784-0},
	doi = {10.1007/s40264-018-00784-0},
	language = {en},
	number = {6},
	urldate = {2024-06-12},
	journal = {Drug Safety},
	author = {Schachterle, Stephen E. and Hurley, Sharon and Liu, Qing and Petronis, Kenneth R. and Bate, Andrew},
	month = jun,
	year = {2019},
	pages = {727--741},
}

@article{park:2022,
	title = {A tree-based scan statistic for zero-inflated count data in post-market drug safety surveillance},
	volume = {12},
	issn = {2045-2322},
	url = {https://www.nature.com/articles/s41598-022-19998-5},
	doi = {10.1038/s41598-022-19998-5},
	language = {en},
	number = {1},
	urldate = {2024-06-12},
	journal = {Scientific Reports},
	author = {Park, Goeun and Jung, Inkyung},
	month = sep,
	year = {2022},
	pages = {16299}
}

@article{heo:2023,
	title = {Signal detection statistics of adverse drug events in hierarchical structure for matched case–control data},
	copyright = {https://creativecommons.org/licenses/by-nc/4.0/},
	issn = {1465-4644, 1468-4357},
	url = {https://academic.oup.com/biostatistics/advance-article/doi/10.1093/biostatistics/kxad029/7330642},
	doi = {10.1093/biostatistics/kxad029},
	language = {en},
	urldate = {2024-06-12},
	journal = {Biostatistics},
	author = {Heo, Seok-Jae and Jeong, Sohee and Jung, Dagyeom and Jung, Inkyung},
	month = oct,
	year = {2023},
	pages = {kxad029},
}

@article{annesi:1989,
	title = {Efficiency of the logistic regression and {C}ox proportional hazards models in longitudinal studies},
	volume = {8},
	copyright = {http://onlinelibrary.wiley.com/termsAndConditions\#vor},
	issn = {0277-6715, 1097-0258},
	url = {https://onlinelibrary.wiley.com/doi/10.1002/sim.4780081211},
	doi = {10.1002/sim.4780081211},
	language = {en},
	number = {12},
	urldate = {2024-06-12},
	journal = {Statistics in Medicine},
	author = {Annesi, Isabella and Moreau, Thierry and Lellouch, Joseph},
	month = dec,
	year = {1989},
	pages = {1515--1521},
}

@article{pottegard:2022,
	title = {Core concepts in pharmacoepidemiology: {Fundamentals} of the cohort and case–control study designs},
	volume = {31},
	issn = {1053-8569, 1099-1557},
	shorttitle = {Core concepts in pharmacoepidemiology},
	url = {https://onlinelibrary.wiley.com/doi/10.1002/pds.5482},
	doi = {10.1002/pds.5482},
	language = {en},
	number = {8},
	urldate = {2024-06-12},
	journal = {Pharmacoepidemiology and Drug Safety},
	author = {Pottegård, Anton},
	month = aug,
	year = {2022},
	pages = {817--826}
}

@article{gerhard:2008,
	title = {Bias: {Considerations} for research practice},
	volume = {65},
	issn = {1079-2082, 1535-2900},
	shorttitle = {Bias},
	url = {https://academic.oup.com/ajhp/article/65/22/2159/5128125},
	doi = {10.2146/ajhp070369},
	language = {en},
	number = {22},
	urldate = {2024-06-12},
	journal = {American Journal of Health-System Pharmacy},
	author = {Gerhard, Tobias},
	month = nov,
	year = {2008},
	pages = {2159--2168},
}

@article{suissa:2008,
	title = {Immortal {Time} {Bias} in {Pharmacoepidemiology}},
	volume = {167},
	issn = {0002-9262, 1476-6256},
	url = {https://academic.oup.com/aje/article-lookup/doi/10.1093/aje/kwm324},
	doi = {10.1093/aje/kwm324},
	language = {en},
	number = {4},
	urldate = {2024-06-12},
	journal = {American Journal of Epidemiology},
	author = {Suissa, S.},
	month = jan,
	year = {2008},
	pages = {492--499}
}

@article{thuy:2024,
    author = {Thai, Thuy N and Winterstein, Almut G and Suarez, Elizabeth A and He, Jiwei and Zhao, Yueqin and Zhang, Di and Stojanovic, Danijela and Liedtka, Jane and Anderson, Abby and Hernández-Muñoz, José J and Munoz, Monica and Liu, Wei and Dashevsky, Inna and Messenger-Jones, Elizabeth and Siranosian, Elizabeth and Maro, Judith C},
    title = "{Triple challenges – Small sample size in both exposure and control groups to scan rare maternal outcomes in a signal identification approach: A simulation study}",
    journal = {American Journal of Epidemiology},
    pages = {kwae151},
    year = {2024},
    month = {06},
    issn = {0002-9262},
    doi = {10.1093/aje/kwae151},
    url = {https://doi.org/10.1093/aje/kwae151},
    eprint = {https://academic.oup.com/aje/advance-article-pdf/doi/10.1093/aje/kwae151/58320103/kwae151.pdf},
}

@article{yih:2023a,
	title = {Safety signal identification for {COVID}-19 bivalent booster vaccination using tree-based scan statistics in the {Vaccine} {Safety} {Datalink}},
	volume = {41},
	issn = {0264410X},
	url = {https://linkinghub.elsevier.com/retrieve/pii/S0264410X2300823X},
	doi = {10.1016/j.vaccine.2023.07.010},
	language = {en},
	number = {36},
	urldate = {2025-06-23},
	journal = {Vaccine},
	author = {Yih, W Katherine and Daley, Matthew F and Duffy, Jonathan and Fireman, Bruce and McClure, David L. and Nelson, Jennifer C. and Qian, Lei and Smith, Ning and Vazquez-Benitez, Gabriela and Weintraub, Eric and Williams, Joshua T.B. and Xu, Stanley and Maro, Judith C.},
	month = aug,
	year = {2023},
	pages = {5265--5270},
}

@article{yih:2023b,
	title = {Sequential {Data}-{Mining} for {Adverse} {Events} {After} {Recombinant} {Herpes} {Zoster} {Vaccination} {Using} the {Tree}-{Based} {Scan} {Statistic}},
	volume = {192},
	copyright = {https://academic.oup.com/journals/pages/open\_access/funder\_policies/chorus/standard\_publication\_model},
	issn = {0002-9262, 1476-6256},
	url = {https://academic.oup.com/aje/article/192/2/276/6760288},
	doi = {10.1093/aje/kwac176},
	language = {en},
	number = {2},
	urldate = {2025-06-23},
	journal = {American Journal of Epidemiology},
	author = {Yih, W Katherine and Kulldorff, Martin and Dashevsky, Inna and Maro, Judith C},
	month = feb,
	year = {2023},
	pages = {276--282},
}

@article{russo:2025,
    author = {Russo, Massimiliano and Sreedhara, Sushama Kattinakere and Smith, Joshua and Davis, Sharon E and Maro, Judith C and Deramus, Thomas and Lii, Joyce and Yang, Jie and Desai, Rishi J and Hernández-Muñoz, José J and Ma, Yong and Wang, Youjin and Jones, Jamal T and Wang, Shirley V},
    title = {Electronic health record-enhanced signal detection using tree-based scan statistic methods},
    journal = {American Journal of Epidemiology},
    volume = {195},
    number = {1},
    pages = {178-187},
    year = {2025},
    month = {09},
    issn = {0002-9262},
    doi = {10.1093/aje/kwaf199},
    url = {https://doi.org/10.1093/aje/kwaf199},
    eprint = {https://academic.oup.com/aje/article-pdf/195/1/178/64219178/kwaf199.pdf}
}

@article{desai:2024,
  title={The FDA Sentinel Real World Evidence Data Enterprise (RWE-DE)},
  author={Desai, Rishi J and Marsolo, Keith and Smith, Joshua and Carrell, David and Penfold, Robert and Pillai, Haritha S and Lii, Joyce and Ngan, Kerry and Winter, Robert and Adgent, Margaret and others},
  journal={Pharmacoepidemiology and drug safety},
  volume={33},
  number={10},
  pages={e70028},
  year={2024},
  publisher={Wiley Online Library}
}

@article{desai:2017,
  title={A propensity-score-based fine stratification approach for confounding adjustment when exposure is infrequent},
  author={Desai, Rishi J and Rothman, Kenneth J and Bateman, Brian T and Hernandez-Diaz, Sonia and Huybrechts, Krista F},
  journal={Epidemiology},
  volume={28},
  number={2},
  pages={249--257},
  year={2017},
  publisher={LWW}
}

@article{smith:2026,
    author = {Smith, Joshua C and Davis, Sharon E and Reeves, Ruth M and Winter, Robert and Whitaker, Jill and Park, Daniel and Wang, Shirley V and Russo, Massimiliano and Maro, Judith C and Hernández-Muñoz, José J and Ma, Yong and Wang, Youjin and Jones, Jamal T and Desai, Rishi J and Matheny, Michael E},
    title = {Characterization and comparison of structured and unstructured electronic health record data mapped to MedDRA for post-marketing surveillance},
    journal = {JAMIA Open},
    volume = {9},
    number = {2},
    pages = {ooag025},
    year = {2026},
    month = {03},
    issn = {2574-2531},
    doi = {10.1093/jamiaopen/ooag025},
    url = {https://doi.org/10.1093/jamiaopen/ooag025},
    eprint = {https://academic.oup.com/jamiaopen/article-pdf/9/2/ooag025/67341224/ooag025.pdf},
}

@article{mishriky:2015,
title = {The efficacy and safety of DPP4 inhibitors compared to sulfonylureas as add-on therapy to metformin in patients with Type 2 diabetes: A systematic review and meta-analysis},
journal = {Diabetes Research and Clinical Practice},
volume = {109},
number = {2},
pages = {378-388},
year = {2015},
issn = {0168-8227},
doi = {https://doi.org/10.1016/j.diabres.2015.05.025},
url = {https://www.sciencedirect.com/science/article/pii/S0168822715002557},
author = {Basem M. Mishriky and Doyle M. Cummings and Robert J. Tanenberg}
}

\end{document}